\documentclass[a4paper,USenglish,cleveref,nameinlink,autoref,thm-restate]{lipics-v2021}

\hideLIPIcs

\usepackage{thmtools}
\usepackage{todonotes}
\usepackage{graphicx} % Required for inserting images
\usepackage[inline]{enumitem}
\usepackage{caption}
\usepackage{xcolor}
\usepackage{xspace}
\usepackage{todonotes}
\usepackage{complexity}

\newtheorem{property}{Property}

\Crefname{observation}{Observation}{Observations}
\Crefname{algorithm}{Algorithm}{Algorithms}
\Crefname{algocf}{Algorithm}{Algorithms}
\Crefname{section}{Section}{Sections}
\Crefname{lemma}{Lemma}{Lemmas}
\Crefname{note}{Note}{Notes}
\Crefname{claim}{Claim}{Claims}
\Crefname{property}{Property}{Properties}
\Crefname{enumi}{Property}{Properties}
\Crefname{figure}{Fig.}{Figs.}
\Crefname{fact}{Fact}{Facts}

\funding{{\em Bekos} and {\em Symvonis} were supported by HFRI Grant No: 26320.}
\title{Internally-Convex Drawings of\\ Outerplanar Graphs in Small Area}
\titlerunning{Internally-Convex Drawings of Outerplanar Graphs in Small Area}

\author{Michael A. Bekos}{University of Ioannina, Greece}{bekos@uoi.gr}{https://orcid.org/0000-0002-3414-7444}{}%{Supported by HFRI Grant No: 26320.}
  
\author{Giordano {Da Lozzo}}{Roma Tre University, Italy
\and
\url{https://uniroma3.gitlab.io/compunet/gdl/}
}{giordano.dalozzo@uniroma3.it}{http://orcid.org/0000-0003-2396-5174}{}%{Supported by the European Union, Next Generation EU, Mission~4, Component 1, CUP J53D23007130006 PRIN proj.\ 2022ME9Z78 ``NextGRAAL: Next-generation algorithms for constrained GRAph visuALization''.}

\author{Fabrizio Frati}{Roma Tre University, Italy}{fabrizio.frati@uniroma3.it}{https://orcid.org/0000-0001-5987-8713}{}%{Supported by the European Union, Next Generation EU, Mission~4, Component 1, CUP J53D23007130006 PRIN proj.\ 2022ME9Z78 ``NextGRAAL: Next-generation algorithms for constrained GRAph visuALization''.}

\author{Giuseppe Liotta}{University of Perugia, Italy}{giuseppe.liotta@unipg.it}{https://orcid.org/0000-0002-2886-9694}{}%{Supported by MUR PON Proj. ARS01 00540 and by MUR PRIN Project no. 2022TS4Y3N}

\author{Antonios Symvonis}{National Technical University of Athens}{symvonis@math.ntua.gr}{https://orcid.org/0000-0002-0280-741X}{}%${Supported by HFRI Grant No: 26320.}

\authorrunning{Bekos, Da Lozzo, Frati, Liotta, Symvonis}

\Copyright{M.\ A.\ Bekos, G.\ Da Lozzo, F.\ Frati, G.\ Liotta, A.\ Symvonis}

\nolinenumbers

\definecolor{lightblue}{HTML}{D5FFFF}
\definecolor{lipicsblue}{rgb}{0.08235294118,0.3098039216,0.537254902}
\definecolor{linkblue}{rgb}{0.098,0.098,0.4392}
\definecolor{ourgreen}{rgb}{0.509,0.745,0.235}
\definecolor{indianred}{rgb}{0.804,0.361,0.361}
\definecolor{indianred1}{rgb}{1,0.416,0.416}
\definecolor{indianred3}{rgb}{0.804,0.333,0.333}
\definecolor{orangered}{rgb}{1,0.271,0}
\definecolor{coral1}{rgb}{1,0.447,0.337}
\definecolor{rosybrown2}{rgb}{0.933,0.231,0.231}
\definecolor{aquamarine4}{rgb}{0.271,0.545,0.455}
\definecolor{chartreuse3}{rgb}{0.4,0.804,0}
\definecolor{mediumpurple3}{rgb}{0.537,0.408,0.804}
\definecolor{mediumvioletred}{rgb}{0.78,0.082, 0.522}

\newcommand{\hly}[1]{\sethlcolor{yellow}\hl{#1}\sethlcolor{yellow}}

\relatedversion{A preliminary version of the paper appears in the proceedings of the 33rd International Symposium on Graph Drawing and Network Visualization (GD 2025) \cite{DBLP:conf/gd/BekosLFLS25}. \cref{sse:diameter} and in particular \cref{th:lower-diameter,th:upper-diameter,th:upper-diameter-outer1} are new to this version.}

\definecolor{defblue}{rgb}{0.121,0.47,0.705}
\DeclareTextFontCommand{\emph}{\color{defblue}\em}

\hypersetup{colorlinks=true,
    linkcolor=aquamarine4,
    anchorcolor=lipicsblue,
    citecolor=lipicsblue,
    filecolor=lipicsblue,
    menucolor=lipicsblue,
    urlcolor=lipicsblue,
    bookmarksopen=true,
    bookmarksopenlevel=2,
    bookmarksnumbered=true,
    plainpages=false,
    }

\EventEditors{Vida Dujmovi\'c and Fabrizio Montecchiani}
\EventNoEds{2}
\EventLongTitle{33rd International Symposium on Graph Drawing and Network Visualization (GD 2025)}
\EventShortTitle{GD 2025}
\EventAcronym{GD}
\EventYear{2025}
\EventDate{September 24--26, 2025}
\EventLocation{Norrk\"{o}ping, Sweden}
\EventLogo{}
\SeriesVolume{357}
\ArticleNo{16}



\begin{document}

\maketitle
\begin{abstract}
A well-known result by Kant [{\em Algorithmica}, 1996] implies that  $n$-vertex outerplane graphs admit embedding-preserving planar straight-line grid drawings where the internal faces are convex polygons in $O(n^{2})$ area. In this paper, we present an algorithm to compute such drawings in $O(n^{1.5})$ area. We also consider outerplanar drawings in which the internal faces are required to be strictly-convex polygons. In this setting, we provide a $\Theta(nk^2)$ area bound for $n$-vertex outerplanar graphs whose weak dual is a path and whose maximum face size is $k$ and a $\Theta(nd^2)$ area bound for $n$-vertex outerplanar graphs whose diameter is bounded by $d$.
\end{abstract}

%consider outerplanar graphs whose  and give a drawing algorithm that achieves , where $k$ is the maximum size of an internal facial~cycle. 

%\newpage

\section{Introduction}
\label{sec:introduction}

Drawing outerplanar graphs in small area is a core topic in Graph Drawing which has been the subject of intense work. The first one we are aware of dates back to 1980, when Leiserson provided an algorithm for the construction of linear-area VLSI layouts of bounded-degree outerplanar graphs~\cite{DBLP:conf/focs/Leiserson80}. In the early 90's, de Fraysseix, Pach, and Pollack~\cite{DBLP:journals/combinatorica/FraysseixPP90} and Schnyder~\cite{DBLP:conf/soda/Schnyder90} proved that $n$-vertex planar graphs, and thus $n$-vertex outerplanar graphs, admit $O(n^2)$-area planar straight-line grid drawings. This bound is the best possible for planar graphs, however whether a subquadratic area upper bound can be achieved for outerplanar graphs has been an intriguing question asked repeatedly (see, e.g.,~\cite{DBLP:conf/gd/Biedl02,DBLP:conf/gd/BrandenburgEGKLM03,DBLP:journals/jgaa/FelsnerLW03,DBLP:journals/dam/GargR07}). A positive answer was first provided by Di Battista and Frati~\cite{DBLP:journals/algorithmica/BattistaF09}. They presented an algorithm which, given an $n$-vertex maximal outerplanar graph $G$ with weak dual tree $T$, uses a structural decomposition of binary trees by Chan~\cite{DBLP:journals/algorithmica/Chan02} in order to construct an $O(n^{1.48})$ area ``star-shaped'' drawing of~$T$; this is a planar straight-line grid drawing which can be augmented to an outerplanar straight-line drawing of $G$ within the same area bound. The same strategy has later led to an almost-linear area bound, namely  $O\left(n2^{\sqrt{2 \log n}} \sqrt{\log n}\right)$~\cite{DBLP:journals/jcss/FratiPR20}, and also to an $O(\Delta n\log n)$ bound for outerplanar graphs of maximum degree $\Delta$~\cite{DBLP:journals/comgeo/Frati12}. 

In this paper we study the area requirements of {\em convex drawings} of outerplanar graphs, in which the boundary of every face (including the outer face) is a convex polygon. It was surprising to us that this question has not been studied before, considering the prominent place that convex drawings occupy in the graph drawing literature. We only mention here the result on convex drawings that is most relevant to our paper. Namely, Chrobak and Kant~\cite{DBLP:journals/ijcga/ChrobakK97,DBLP:journals/algorithmica/Kant96} and, independently, Felsner~\cite{DBLP:journals/order/Felsner01} have presented extensions of the algorithms by de Fraysseix, Pach, and Pollack~\cite{DBLP:journals/combinatorica/FraysseixPP90} and Schnyder~\cite{DBLP:conf/soda/Schnyder90}, proving that every $3$-connected planar graph admits a convex drawing in quadratic area, which is asymptotically optimal. Specifically, Felsner's upper bound is $(f-1)\times(f-1)$, where $f$ denotes the number of faces. 

{\em Strictly-convex drawings}, in which no three vertices on the boundary of a face are allowed to be collinear, require larger area. Indeed, even an $n$-vertex cycle cannot be drawn as a strictly-convex polygon on a grid of size $o(n^3)$~\cite{Andrews1963}. For $n$-vertex $3$-connected planar graphs, this bound is not yet matched by a corresponding upper bound. The best known upper bound is due to B{\'{a}}r{\'{a}}ny and Rote~\cite{Barany2006}, who proved that $n$-vertex $3$-connected planar graphs admit strictly-convex grid drawings in $O(n^4)$ area. Bekos, Gronemann, Montecchiani, and Symvonis~\cite{DBLP:journals/jocg/BekosGMS24} provided an alternative algorithm with the same asymptotic area bound but with significantly smaller hidden constants. 

The construction of small-area convex grid drawings of outerplanar graphs encounters two natural barriers. First, a planar graph admits a convex drawing only if it is biconnected, hence one is required to consider biconnected outerplanar graphs only. Second, as we show later, there exist $n$-vertex biconnected outerplanar graphs that require $\Omega(n^3)$ area in any (not necessarily strictly-)convex grid drawing; note that an $O(n^3)$ upper bound for (even strictly-)convex drawings of $n$-vertex biconnected outerplanar graphs can be achieved by drawing the cycle delimiting the outer face as a strictly-convex polygon~\cite{Andrews1963} and drawing the internal edges as chords. We thus relax the convexity requirement and require only the internal faces to be convex, that is, we consider {\em internally-convex drawings}. For outerplanar graphs, this is a natural variation of the convex drawing style, taking into account that the outer face of an outerplanar graph is {\em special}: All the vertices are incident to it and this property defines the graph class. Chrobak and Kant's algorithms~\cite{DBLP:journals/ijcga/ChrobakK97,DBLP:journals/algorithmica/Kant96} allow us to obtain internally-convex drawings of $n$-vertex outerplanar graphs on a grid of size $O(n^2)$ (via a suitable augmentation of the input graph to a 3-connected planar graph, by adding vertices and edges in its outer face). This restores in the internally-convex setting the pursuit for a sub-quadratic area upper bound for grid drawings of outerplanar graphs. We also consider the, in our opinion equally interesting, question of determining the area requirements of {\em internally-strictly-convex drawings} of outerplanar~graphs. 

\medskip\noindent\textbf{Our contributions.} Our main result is an algorithm that constructs an internally-convex grid drawing with $O(n^{1.5})$ area for $n$-vertex outerplanar graphs. We employ some ideas and tools from a recent paper by Biedl, Liotta, Lynch, and Montecchiani~\cite{BiedlLLM24}. The main geometric component of our algorithm consists of constructing a ``leveled'' drawing, where consecutive levels are drawn on consecutive horizontal grid lines, up to a level where there are sufficiently few vertices so that the drawing can ``make a turn'' and extend horizontally rather than vertically. We remark that ours is the first algorithm that constructs outerplanar straight-line grid drawings in sub-quadratic area by drawing directly the outerplanar graph, without passing through a star-shaped drawing of its weak dual tree. This aspect might be of independent interest.

For internally-strictly-convex grid drawings of $n$-vertex outerplanar graphs, the same arguments as in the strictly-convex setting prove that $\Theta(n^3)$ is a tight bound for the area requirements. However, the natural question to consider here is how the area relates not only to the number $n$ of vertices but also to other parameters of the graph, like the maximum size~$k$ of an internal face or the diameter $d$. We show that there exist outerplanar graphs requiring~$\Omega(nk^2)$ area and outerplanar graphs requiring~$\Omega(nd^2)$ area in any internally-strictly-convex grid drawing. We also show that the latter lower bound can be matched by a corresponding~$O(nd^2)$ upper bound for all the $n$-vertex outerplanar graphs and that the former lower bound can be matched by a corresponding~$O(nk^2)$ upper bound for the $n$-vertex outerpaths, i.e., the biconnected outerplanar graphs whose weak dual is a path. Our upper bounds for the area required by internally-strictly-convex grid drawings have implications for the construction of polynomial-area straight-line drawings of memningful families of outer-$1$-planar graphs, i.e. those graphs that admit a drawing where all vertices are incident to the outer face and every edge is crossed at most once. We recall that, while deciding whether a graph $G$ admits a drawing where each edge is crossed at most once is \NP-hard~\cite{grigoriev2007algorithms,korzhik2013minimal}, testing whether $G$ is outer-$1$-planar can be performed in $O(n)$-time~\cite{DBLP:journals/algorithmica/HongEKLSS15}. To the best of our knowledge, no-polynomial area straight-line drawing algorithms are currently known for general outer-$1$-planar graphs.  Instead, Dehkordi and Eades~\cite{dehkordi2012every} describe an exponential-area straight-line drawing algorithm for outer-$1$-planar graphs, in which all edge crossings form right angles.

The rest of the paper is organized as follows. In \cref{sec:preliminaries}, we introduce some definitions and prove a lower bound for the area requirements of convex drawings. In \cref{se:non-strict} and \cref{se:strict}, we present our results for internally-convex  and internally-strictly-convex drawings, respectively. Finally, we conclude in \cref{sec:conclusions} with several open problems.

% Full proofs of statements marked with a ($\star$) can be found in the extended version of the paper~\cite{thispaper}.
% \todo[inline]{Fix this bibentry.}

\section{Preliminaries}
\label{sec:preliminaries}

We assume familiarity with Graph Drawing, see, e.g.~\cite{DBLP:books/ph/BattistaETT99}. 

A \emph{drawing} of a graph maps each vertex to a distinct point in the plane and each edge to a Jordan arc between its endpoints. In a \emph{grid drawing} vertices have integer coordinates.
Let~$\Gamma$ be a grid drawing of a graph. Let $\max_q(\Gamma)$ and $\min_q(\Gamma)$ be the maximum and minimum $q$-coordinate of a vertex in $\Gamma$, with $q \in \{x,y\}$, respectively. The \emph{height} $h(\Gamma)$ and \emph{width} $w(\Gamma)$ of~$\Gamma$ are the positive integers $h(\Gamma) = \max_y(\Gamma)-\min_y(\Gamma) + 1$ and $w(\Gamma) = \max_x(\Gamma)-\min_x(\Gamma) + 1$, respectively. In other words, $h(\Gamma)$ and $w(\Gamma)$ represent the number of horizontal and vertical grid lines intersecting the smallest axis-parallel rectangle enclosing $\Gamma$, respectively.
The \emph{area} of a drawing is the product of its width and height.

A drawing is \emph{planar} if no two edges intersect, except at a common end-point. In a planar drawing, a vertex or an edge is \emph{external} if it is incident to the outer face, and \emph{internal} otherwise. A drawing is \emph{outerplanar} if it is planar and every vertex is external. A graph is \emph{planar} if it admits a planar drawing and \emph{outerplanar} if it admits an outerplanar drawing. A planar drawing partitions the plane into topologically connected regions called \emph{faces}.
The unbounded face is the \emph{outer face}, whereas the bounded faces are the \emph{internal faces}.
Two planar drawings of a connected graph are \emph{equivalent} if they induce (i) the same counter-clockwise order of the edges incident to each vertex and (ii) the same counter-clockwise order of the edges along the boundaries of their outer faces. The equivalence classes of equivalent outerplanar drawings are called \emph{outerplanar embeddings}, or in this paper just \emph{embeddings}. An \emph{outerplane graph} is a planar graph equipped with an outerplanar embedding. A biconnected outerplanar graph has a unique outerplanar embedding (up to an inversion of all the orders defining the embedding), hence we often talk about faces of a biconnected outerplane graph, referring to faces of its  unique outerplanar embedding. Also, a drawing of a graph is \emph{embedding-preserving} if it respects an embedding associated to the graph.

In a \emph{straight-line drawing}, edges are drawn as straight-line segments. In a planar straight-line drawing, a face is \emph{convex} (\emph{strictly-convex}) if it is delimited by a polygon whose internal angles are not larger than $180^\circ$ (resp.\ are smaller than $180^\circ$). A \emph{convex drawing} (a \emph{strictly-convex drawing}) is a straight-line planar drawing in which every face, including the outer face, is convex (resp., strictly-convex). An \emph{internally-convex drawing} (an \emph{internally-strictly-convex drawing}) is a straight-line planar drawing in which every internal face, but not necessarily the outer face, is convex (resp., strictly-convex). %
When dealing with outerplanar graphs, we do not require convex drawings to be outerplanar per se. Our results however hold in the strongest possible sense: The lower bounds on the area hold true even for drawings that are not outerplanar, while the upper bounds are proved by constructing outerplanar drawings.

The following theorem motivates the study of internally-convex drawings.

%In this paper, we focus on internally-convex grid drawings and aim at computing such drawings using as little area as possible. 
%\hly{In particular,} we aim at constructing internally-convex drawing that use {\em subquadratic area}.
%\hlg{TODO: Argue that $O(n^2)$-area internally-convex drawings of outerplanar graphs can be obtained by adding a universal vertex and using Kant's algorithm.}
%There are two main reasons to look at internally-convex drawings of outerplanar graphs. The first reason is that outerplanar graphs that are not biconnected do not admit convex drawings. The second reason is stated in the following theorem. \hly{Gio: should't we clarify that we aim at subquadratic area before?}

\begin{theorem}\label{th:convex-lower} 
There exists an $n$-vertex biconnected outerplanar graph that requires $\Omega(n^3)$ area in every convex grid drawing.
\end{theorem}
\begin{proof}
Let $G$ be the $n$-vertex biconnected outerplanar graph obtained from a cycle $\mathcal C$ with~$n/2$ vertices by attaching a degree-$2$ vertex incident to the end-vertices of each edge of~$\mathcal C$. First, every convex drawing of~$G$ respects the unique outerplanar embedding of~$G$, as moving a degree-$2$ vertex inside $\mathcal C$ would create an angle larger than $180^\circ$ in an internal face at that vertex. Second, in every outerplanar convex drawing of~$G$, the drawing of~$\mathcal C$ is strictly-convex, as the angle at each vertex $v$ of~$\mathcal C$ in the interior of the cycle delimiting the outer face of~$G$ is at most $180^\circ$ and part of this angle is used by the two triangular faces incident to $v$. Since every strictly-convex drawing of a cycle requires cubic area~\cite{Andrews1963}, the lower bound follows.
\end{proof}

\begin{figure}[t!]
%  \begin{center}
\centering
\includegraphics[page=1,width=.6\textwidth]{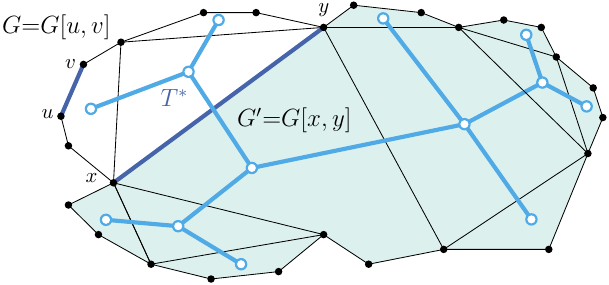}
%  \end{center}
 % \vspace{-5mm}
  \caption{
  An outerplanar graph $G[u,v]$ rooted at the external edge $(u,v)$ and its weak dual $T^*$. The outerplanar graph $G' = G[x,y]$ is light-blue shaded.\label{fig:def-rooted}}
\end{figure}

Let $G$ be a biconnected outerplane graph with $n\geq 3$ vertices; refer to \cref{fig:def-rooted}. Note that $G$ has at least one internal face. 
The \emph{weak dual} of~$G$, denoted by $T^*$, is the tree having a vertex for each internal face of~$G$ and an edge between any two vertices corresponding to internal faces sharing an edge. 
For simplicity, we use the same label to denote a vertex of~$T^*$ and the corresponding face of~$G$. 
If $T^*$ is a path, then $G$ is an \emph{outerpath}. 
Let $(u,v)$ be an external edge of~$G$ such that $v$ immediately precedes $u$ in counter-clockwise order along the outer face of~$G$. The outerplane graph $G$ with the designated edge $(u,v)$, denoted by $G[u,v]$, is said to be \emph{rooted at $(u,v)$}.  Let $(x,y)$ be an internal edge of~$G[u,v]$ such that $u,x,y,v$ appear in this counter-clockwise order along the outer face of~$G$ (possibly $u=x$ or $y=v$ may hold). Consider the outerplane subgraph $G'$ of~$G$ induced by the vertices $x$, $y$, and the vertices that are encountered when traversing the boundary of the outer face of~$G$ in counter-clockwise direction from~$x$ to~$y$. 
Observe that the edge $(x,y)$ is an external edge of~$G'$ such that $y$ immediately precedes $x$ in counter-clockwise order along the outer face of~$G'$. Then we denote by $G[x,y]$ the outerplanar graph $G'$ rooted at the edge $(x,y)$.

Let $G$ be rooted at $(u,v)$ and let $f_1$ be the internal face of~$G$ incident to $(u,v)$. 
The \emph{extended weak dual tree} $T$ of~$G[u,v]$ is the tree with root $f_1$ obtained from the weak-dual tree $T^*$ of $G$ by inserting a new leaf $\ell_e$ for each external edge $e$ of~$G$ different from $(u,v)$ and an edge $(f,\ell_e)$ if $f$ is the internal face of~$G$ incident to $e$. Let $d$ be the number of vertices of~$T$. Note that $d\geq n$, as $T$ contains the root $f_1$ plus one leaf for each of the $n-1$ external edges of~$G$ different from $(u,v)$. Also, $d\leq 2n-3$, since $G$ has $n-1$ external edges different from $(u,v)$ and at most $n-2$ internal faces. 
Let $\pi=(f_1,f_2,\dots,f_p)$ be a root-to-leaf path in $T$. We denote by $G[\pi]$ the outerpath induced by the vertices incident to the faces $f_1,f_2,\dots,f_{p-1}$. We call $G[\pi]$ the \emph{outerpath dual to $\pi$}. Denote by $e_p$ the edge of~$G$ dual to $(f_{p-1},f_p)$. Let $u=z_1,z_2,\dots,z_a,z_{a+1},\dots,z_m=v$ be the counter-clockwise order of the vertices along the outer face of~$G[\pi]$, where $e_p=(z_a,z_{a+1})$. We call the rooted outerplanar graphs $G[z_1,z_2],\dots,G[z_{a-1},z_a]$ \emph{left subgraphs of~$G$ at $\pi$} and the rooted outerplanar graphs $G[z_{a+1},z_{a+2}],\dots,G[z_{m-1},z_m]$ \emph{right subgraphs of~$G$ at $\pi$}.

In this paper, we also deal with some non-planar drawings, which are introduced below. A drawing of a graph is \emph{outer-$1$-planar} if every vertex is external and every edge is crossed at most once. An \emph{outer-$1$-planar graph} is a graph that admits an outer-$1$-planar drawing. Two outer-$1$-planar drawings $\Gamma$ and $\Gamma'$ of a graph are \emph{equivalent} if: (i) any two edges $e_1$ and $e_2$ cross in $\Gamma$ if and only if they cross in $\Gamma'$; and (ii) the planar drawings $\Gamma_p$ and $\Gamma'_p$ obtained from $\Gamma$ and $\Gamma'$, respectively, by introducing a vertex $v_{1,2}$ in place of the crossing point between any two edges $e_1$ and $e_2$ are equivalent. An \emph{outer-$1$-planar embedding} is an equivalence class of outer-$1$-planar drawings, and an \emph{outer-$1$-plane graph} is an outer-$1$-planar graph equipped with an outer-$1$-planar embedding.

The \emph{outer frame} of an outer-$1$-plane graph~$G$ is defined as follows. Consider an outer-$1$-planar drawing $\Gamma$ of~$G$ respecting the outer-$1$-planar embedding of $G$. For any two edges $e_1=(u_1,v_1)$ and $e_2=(u_2,v_2)$ that cross in a point $c_{1,2}$, insert in $G$ and $\Gamma$ the edges $(u_1,u_2)$, $(u_1,v_2)$, $(v_1,u_2)$, and $(v_1,v_2)$, if they do not already belong to $G$, so that the $4$-cycle $(u_1,u_2,v_1,v_2)$ encloses the edges $e_1$ and $e_2$, and encloses no vertex. This augmentation is always possible and keeps $G$ outer-$1$-plane. Indeed, if $G$ does not contain, say, the edge $(u_1,u_2)$, then it can be inserted  in $\Gamma$ sufficiently close to the polyline $(u_1,c_{1,2},u_2)$, so that it does not cross any edge, and so that every vertex different from $u_1$ and $u_2$ is outside the closed curve composed of $(u_1,u_2)$ and of $(u_1,c_{1,2},u_2)$. If $G$ does contain the edge $(u_1,u_2)$, then again the closed curve composed of the edge $(u_1,u_2)$ and of $(u_1,c_{1,2},u_2)$ does not contain any vertex in its interior, given that $\Gamma$ is outer-$1$-planar, and does not cross any edge, given that $e_1$ and $e_2$ already cross each other. After the augmentation has been performed for every pair of crossing edges, we obtain the \emph{framed augmentation} of $G$. The outer frame of~$G$ is then defined by removing the crossing edges from the framed augmentation of $G$.

\section{Internally-Convex Drawings} \label{se:non-strict}

%Selling points:
%Not every outerplanar graph admits a drawing in which all faces (including the outer one) are convex.
%Biconnected outerplanar graphs admit drawings in which all faces (including the outer one) are convex, but they require $\Omega(n^3)$ area (even just for convexity).

In this section we prove the following theorem, which is our main result. 

\begin{theorem} \label{th:non-strictly}
Every $n$-vertex outerplane graph admits an embedding-preserving internally-convex grid drawing in $O(n^{1.5})$ area.
\end{theorem}

Our proof is inspired by the proof that every outer-1-plane graph admits an embedding-preserving visibility representation in $O(n^{1.5})$ area~\cite{BiedlLLM24}. In particular, we are going to use the following theorem, which generalizes a well-known result by Chan \cite{DBLP:journals/algorithmica/Chan02}. 

\begin{theorem} \label{th:decomposition} (\cite[Theorem 1]{BiedlLLM24})
Let $p=0.48$. Given any rooted tree with $n$ vertices, there exists a root-to-leaf path $\pi$ such that for any left subtree $\alpha$ and for any right subtree $\beta$ of~$\pi$, $|\alpha|^p+|\beta|^p\leq (1-\delta)n^p$, for some constant $0<\delta<1$.
\end{theorem}

The following lemma is similar to the combination of Definition 1 and Observation 3 from \cite{BiedlLLM24}, although with different constants. For the sake of completeness, we provide a proof, which follows the one in \cite{BiedlLLM24} almost verbatim. Recall that $p=0.48$ and $0<\delta<1$ are the constants from \cref{th:decomposition}. We define the constant $c$ as  $c:=2/\delta$. Observe that $c>2$.

\begin{lemma}\label{le:induction_works}  
Let $f(d)$ be the recursive function defined on $\mathbb N_0$ as follows.  First, $f(0)=0$ and $f(1)=1$. Second, for $d\in \mathbb N_0$ with $d\geq 2$, we have that 
\[
f(d)=\max\limits_{d_a,d_b\in \mathbb N_0:d_a^p+d_b^p\leq (1-\delta)d^p} \{f(d_a)+f(d_b)\}+2\sqrt d.
\]
Then $f(d)\leq c\sqrt d$. 
\end{lemma}

\begin{proof}
We prove the statement by induction on $d$. The statement is obviously true if $d\leq 1$, given that $c>2$. Suppose next that $d\geq 2$. Let $d_a$ and $d_b$ be non-negative integers at which the maximum in the definition of~$f(d)$ is achieved, i.e., $d_a^p+d_b^p\leq (1-\delta)d^p$ and $f(d)=f(d_a)+f(d_b)+2\sqrt d$. Since $\delta>0$, we have $d_a<d$ and $d_b<d$. Dividing  $d_a^p+d_b^p\leq (1-\delta)d^p$ by $d^p$ we get $\left(\frac{d_a}{d}\right)^p+\left(\frac{d_b}{d}\right)^p\leq (1-\delta)$. By $d_a<d$, $d_b<d$, and $p<1/2$, we get $\left(\frac{d_a}{d}\right)^{1/2}+\left(\frac{d_b}{d}\right)^{1/2}<\left(\frac{d_a}{d}\right)^p+\left(\frac{d_b}{d}\right)^p$, which implies that $\left(\frac{d_a}{d}\right)^{1/2}+\left(\frac{d_b}{d}\right)^{1/2}<(1-\delta)$, that is, $\sqrt d_a +\sqrt d_b<(1-\delta)\sqrt d$. Therefore, from $f(d)=f(d_a)+f(d_b)+2\sqrt d$ by induction we get $f(d)\leq c\sqrt d_a +c\sqrt d_b +2\sqrt d$, which by the previous inequality is smaller than $c(1-\delta)\sqrt d+2\sqrt d$, which is equal to $c\sqrt d$ for $c=2/\delta$.
\end{proof}

Let $G$ be an $n$-vertex outerplane graph; refer to \cref{fig:decomposition}. In what follows, we assume~$G$ to be biconnected. Indeed, if it is not, it can be augmented to a biconnected outerplane graph~$G'$ by introducing edges in its outer face. Then the restriction of an embedding-preserving internally-convex grid drawing $\Gamma'$ of~$G'$ yields an embedding-preserving internally-convex grid drawing of~$G$ that inherits the area upper bounds of~$\Gamma'$. Let $u$ and $v$ be any two vertices such that $v$ immediately precedes $u$ in counter-clockwise order along the outer face of~$G$. We root $G$ at the edge~$(u,v)$. 
Recall that $G[u,v]$ denotes the outerplanar graph $G$ rooted at the edge~$(u,v)$.

We show that $G[u,v]$ admits a (small-area) \emph{$uv$-separated drawing}. This is an embedding-preserving internally-convex grid drawing $\Gamma$ of~$G$ satisfying the following properties:
\begin{enumerate}[label=\bf {P.}\arabic*,leftmargin=21pt]
    \item\label{prp:1}$u$ and $v$ lie on a horizontal line $h_0$;
    \item\label{prp:2}all neighbors of~$u$ and $v$ lie on the horizontal line $h_1$ one unit below $h_0$;
    \item\label{prp:3}all the other vertices of~$G$ lie not above $h_1$; and
    \item\label{prp:4}all vertices different from $u$ (from $v$) are to the right of $u$ (resp.\ to the left of $v$).
\end{enumerate}

The following property is a direct consequence of the definition of~$uv$-separated drawing.

\begin{figure}
    \centering
    \includegraphics[width=.8\textwidth]{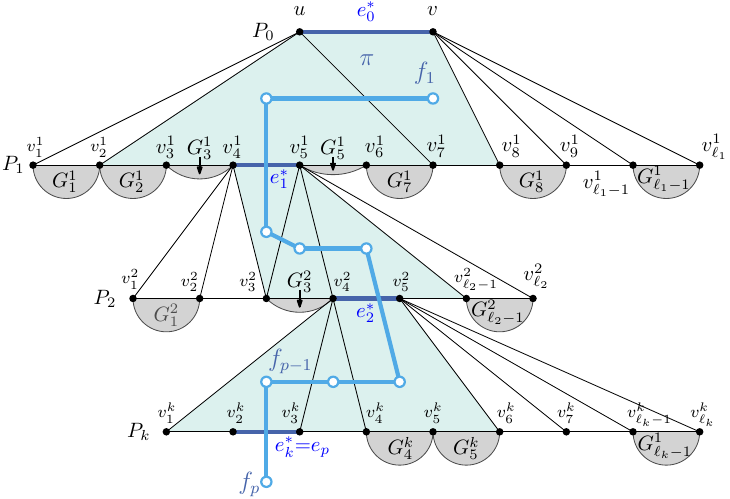}
    \caption{
    An outerplanar graph $G[u,v]$ rooted at the external edge $(u,v)$. The outerpath $G[\pi]$ dual to the path $\pi$ selected by \cref{th:decomposition} is light-blue shaded. Transition edges are drawn~thick~blue.}
    \label{fig:decomposition}
\end{figure}

\begin{property}\label{pro:shift}
Let $\Gamma$ be a $uv$-separated drawing of~$G[u,v]$. Vertices $u$ and $v$ can be shifted arbitrarily by integer amounts to the left and right, respectively, while maintaining $\Gamma$ an outerplanar internally-convex grid drawing.
\end{property}

If $n\geq 3$, then let $T$ be the extended weak dual tree of~$G[u,v]$ and let $d$ be the number of vertices of~$T$. If $n=2$, then $T$ is not defined and we set $d=1$.  This choice is justified by the size of the subtrees of~$T$ we generate. Indeed, if $n\geq 3$, we will select a root-to-leaf path $\pi$ in $T$ and consider left and right subgraphs of~$G$ at $\pi$. If any of such subgraphs is just an edge, it is associated with a subtree of~$T$ which consists of a single vertex, from which the correspondence between $n=2$ and $d=1$ stems. 

We show an algorithm that, by induction on $n$, constructs a $uv$-separated drawing $\Gamma$ of~$G[u,v]$ satisfying the following two properties: 
\begin{itemize}
\item \label{prop:W}\hyperref[prop:W]{\em Property W:} every vertical grid line intersecting $\Gamma$ contains at least one vertex of~$G$ (and thus the width of~$\Gamma$ is at most $n$); and
\item \label{prop:H}\hyperref[prop:H]{\em Property H:} every horizontal grid line intersecting $\Gamma$ contains at least one vertex of~$G$ and the height of~$\Gamma$ is at most $f(d)$.
\end{itemize}
The statement implies \cref{th:non-strictly}, given that $f(d)\in O(\sqrt d)$ by \cref{le:induction_works} and since $d\leq 2n-3$.

Denote by~$h(d)$ the maximum value of~$h(\Gamma)$, over all drawings $\Gamma$ constructed by our algorithm for rooted outerplane graphs whose extended weak dual has $d$ vertices.

In the base case, we have $n=2$ (and thus $d=1$). Hence, $G$ is the edge $(u,v)$ and a drawing $\Gamma$ satisfying \hyperref[prop:W]{\em Properties W} and \hyperref[prop:H]{\em H} is constructed by placing $u$ on a grid point one unit to the left of~$v$. In particular, note that $h(\Gamma)=1$ and that~$f(1)=1$.

In the inductive case, we have $n>2$, hence $G$ has internal faces, since it is biconnected. Let $f_1$ be the root of~$T$, which is the internal face of~$G$ incident to $(u,v)$. Let $\pi=(f_1,f_2,\dots,f_p)$ be the root-to-leaf path in $T$ from \cref{th:decomposition}, where $e_p$ is the edge of~$G$ dual to $(f_{p-1},f_p)$. 

We introduce some notation and definitions. Call $e^*_0=(u,v)$ a \emph{transition edge}; let $V_0:=\{u,v\}$, let the path $P_0=(v^0_0,v^0_1)$ coincide with $(u,v)$, and let the rooted outerplane graph $G^*_0$ coincide with $G[u,v]$. Suppose that, for some $i\geq 1$, a transition edge $e^*_{i-1}$ has been defined that is part of a path $P_{i-1}$. Consider the set of vertices that do not belong to $V_0\cup \dots\cup V_{i-1}$ and that belong to the faces of $G$ incident to the end-vertices of~$e^*_{i-1}$; these vertices induce a path $P_i=(v^i_1,\dots,v^i_{\ell_i})$, where the labels of the vertices are such that $u,v^i_1,v^i_{\ell_i},v$ appear in this counter-clockwise order along the outer face of~$G$. See again \cref{fig:decomposition}, where, for example, $P_2$ is the path induced by the vertices that do not belong to $V_0 \cup V_1$ (roughly speaking, by the vertices that are on the ``correct side'' of $e^*_1$) and that belong to the faces of $G$ incident to the end-vertices of~$e^*_1$ (only one of such faces is incident to both the end-vertices of~$e^*_1$, the other ones are incident to a single end-vertex of~$e^*_1$). Only one of the edges of~$P_i$ might be dual to an edge of~$\pi$; this edge of~$P_i$ is the transition edge $e^*_i$ with end-vertices $v^i_{s_i}$ and $v^i_{s_i+1}$, for some $1\leq s_i\leq \ell_i-1$. The set $V_i$ comprises the vertices $v^i_1,\dots,v^i_{\ell_i}$ of~$P_i$, as well as the vertices of the subgraphs $G^i_j:=G[v^i_j,v^i_{j+1}]$ of~$G$, for $j=1,\dots,\ell_i-1$ with $j\neq s_i$. Observe that, for $j=1,\dots,s_i-1$, the graph $G^i_j$ is a left subgraph of~$G$ at $\pi$ or a proper subgraph of a left subgraph of~$G$ at $\pi$. Similarly, for $j=s_i+1,\dots,\ell_i-1$, the graph $G^i_j$ is a right subgraph of~$G$ at $\pi$ or a proper subgraph of a right subgraph of~$G$ at $\pi$. We denote by $G^*_i$ the graph $G[v^i_{s_i},v^i_{s_i+1}]$. For the last-defined path $P_k$, we have that either the transition edge $e^*_k$ does not exist or, if it does, it is the edge $e_p$, which is incident to the outer face of~$G$ and thus $G^*_k$ coincides with the edge $e^*_k$. An example of a situation in which $e^*_k$ does not exist is the one in which $P_k$ is a single vertex. This happens if $f_{p-1}$, the second-to-last vertex of $\pi$, is a triangular face delimited by the end-vertices of $e^*_{k-1}$ and by one additional vertex, which coincides with $P_k$. In this case $e_p$, the edge of $G$ dual to the last edge $(f_{p-1},f_p)$ of $\pi$, connects a vertex of $e^*_{k-1}$ with the unique vertex of $P_k$ and $e^*_{k}$ does not exist.

%\todo{Beppe: We may want to illustrate these definitions (or part of them) referring to Figure 2}

%The last-defined path $P_k$ does not contain any transition edge; this happens when $f_{p-1}$ is incident to an end-vertex of~$e^*_{k-1}$. 

%\emph{attached to} the edges of~$P_i$ different from $e^*_i$. These subgraphs are formally defined as follows. Let $e^i_j=(v^i_j,v^i_{j+1})$ be an edge of~$P_i$ different from $e^*_i$, with $1 \leq j < \ell_i$. Then the subgraph $G^i_j$ of~$G$ attached to  $e^i_j$ is the graph induced by the vertices encountered when traversing the boundary of the outer face of~$G$ in counter-clockwise direction from $v^i_j$ and $v^i_{j+1}$. It might be the case that $G^i_j$ is just a single edge. The last-defined path $P_k$ does not contain any transition edge; this happens when one of the faces incident to an end-vertex of~$e^*_{k-1}$ is the last vertex of~$\pi$. 

%every edge $e$ of~$P_i$ different from $e^*_i$. Remove from $T$ the edge dual to $e$ and let $T_e$ be the resulting subtree that does not contain the root of~$T$. Then the subgraph $G_e$ attached to $e$ is the dual of~$T_e$, that is, for every vertex $w$ of~$T_e$, the graph $G_e$ contains the edges incident to the face of~$G$ corresponding to $w$. 

Ideally, we would like the $uv$-separated drawing we construct to have the following features. The vertices $v^i_1,\dots,v^i_{\ell_i}$ of each path $P_i$ should appear on a common horizontal line $h_i$ in this left-to-right order, and~$h_i$ should lie one unit above $h_{i+1}$. 
The subgraphs $G^i_j$ with $j\neq s_i$ should be recursively drawn and placed below~$h_i$, except for the edge $(v^i_j,v^i_{j+1})$ which lies on~$h_i$, in such a way that any two of these subgraphs are \emph{horizontally disjoint}, that is, there exists a vertical line that keeps the drawings of the two subgraphs on different sides, except, possibly, for a single vertex shared by the two subgraphs which lies on the line. The subgraph~$G^i_{s_i}$ is not drawn recursively, as it coincides with the subgraph $G^*_i$ on which the level-by-level construction we are describing is further applied. Recall that $k$ is the index of the path $P_k$ defined last in the procedure discussed above.
If $k\leq \sqrt n$ ({\bf Case 1}), this plan can be accomplished. Otherwise ({\bf Case~2}), by a simple vertex counting argument, there exists an index $2\leq t \leq 1+\sqrt n$ such that the set $V_t$ has at most $\sqrt n$ vertices. In this case, we draw the vertices in $V_t$ differently than before, so that the transition edge~$e^*_t$ is vertical. The small size of~$V_t$ allows us to ``turn'' from horizontal to vertical without increasing the height dramatically. From $e^*_t$ towards the end of~$\pi$, the construction proceeds by drawing the vertices of the outerpath dual to $\pi$ on two horizontal lines, and by placing the recursively-constructed drawings of the right and left subgraphs of~$G$ at $\pi$ above and below it, respectively.

\begin{figure}
    \centering
    \includegraphics[page=2,width=.9\textwidth]{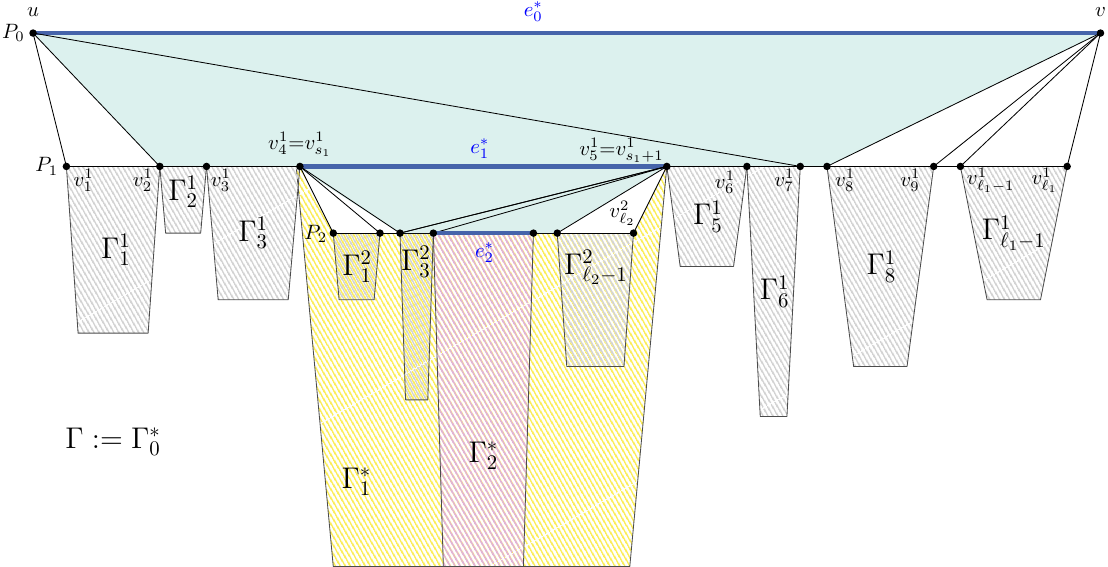}
    \caption{Construction for a $uv$-separated drawing of~$G[u,v]$ when $k \leq \sqrt{n}$. To avoid cluttering, vertical and horizontal proportions are not respected, but vertices are assumed to lie on the grid.}
    \label{fig:simple-case}
\end{figure}

%For $i=0,1,\dots,k-1$, let $e^*_i=(v^i_{s(i)},v^i_{s(i)+1})$ and let $G^*_i$ be the subgraph of~$G$ induced by $\{v^i_{s(i)},v^i_{s(i)+1}\}\cup V_{i+1}\cup V_{i+2}\cup \dots\cup V_k$. 

We now formalize this argument. As sketched above, we distinguish two cases. 
In {\bf Case~1}, we have $k\leq \sqrt n$; refer to \cref{fig:simple-case}. For $i=1,\dots,k$ and for $j=1,\dots,\ell_i-1$ with $j\neq s_i$, we recursively construct a $v^i_jv^i_{j+1}$-separated drawing $\Gamma^i_j$ of~$G^i_j[v^i_j,v^i_{j+1}]$. 
We construct, for $i=k-1,k-2,\dots,0$, a  $v^i_{s_i}v^i_{s_i+1}$-separated drawing $\Gamma^*_i$ of~$G^*_i[v^i_{s_i},v^i_{s_i+1}]$ as described in the following. 
Notice that the order of the indices ensures that, if $G^*_{i+1}$ exists, then its drawing~$\Gamma^*_{i+1}$ has been already constructed once $\Gamma^*_i$ is about to be constructed. In order to start the iteration, if $e^*_k$ does not exist, then $G^*_k$ does not exist either, and hence $\Gamma^*_k$ does not need to be defined. If $e^*_k$ exists, then $G^*_k$ coincides with such an edge and $\Gamma^*_k$ is constructed by placing $v^k_{s_k}$ on a grid point one unit to the left of~$v^k_{s_k+1}$. In order to construct $\Gamma^*_i$, we place the recursively constructed drawings $\Gamma^{i+1}_j$ with $j=1,\dots,\ell_{i+1}-1$ with $j\neq s_{i+1}$, as well as the drawing $\Gamma^*_{i+1}$ (if $G^*_{i+1}$ exists), side by side, so that the placement of a vertex $v^{i+1}_j$ coincides in two drawings it belongs to; these are $\Gamma^{i+1}_{j-1}$ and $\Gamma^{i+1}_j$ if $j\notin \{s_{i+1},s_{i+1}+1\}$,
%$v^{i+1}_j$ is neither $v^{i+1}_{s_{i+1}}$ nor $v^{i+1}_{s_{i+1}+1}$
or are $\Gamma^{i+1}_{s_{i+1}-1}$ and $\Gamma^*_{i+1}$ if 
$j=s_{i+1}$,
%$v^{i+1}_j$ is $v^{i+1}_{s_{i+1}}$, 
or are $\Gamma^*_{i+1}$ and $\Gamma^{i+1}_{s_{i+1}+1}$ if 
$j=s_{i+1}+1$.
%$v^{i+1}_j$ is $v^{i+1}_{s_{i+1}+1}$. 
This ensures that the vertices of~$P_{i+1}$ all lie on the same horizontal line $h_{i+1}$. Further, we place $v^i_{s_{i}}$ one unit above and one unit to the left of the leftmost vertex of~$P_{i+1}$, and $v^i_{s_{i}+1}$ one unit above and one unit to the right of the rightmost vertex of~$P_{i+1}$. This results in the desired $v^i_{s_{i}}v^i_{s_{i}+1}$-separated drawing $\Gamma^*_i$ of~$G[v^i_{s_i},v^i_{s_i+1}]$. When $i=0$, the resulting drawing $\Gamma:=\Gamma^*_0$ is a $uv$-separated drawing of~$G[u,v]$, as established by the following lemma. 

\begin{lemma}\label{le:Case1}
In Case 1, the constructed drawing $\Gamma$ is a $uv$-separated drawing of~$G[u,v]$ satisfying 
\hyperref[prop:W]{Properties W} and \hyperref[prop:H]{H}.
\end{lemma}

% \begin{restatable}[{\hyperref[le:Case1]{$\star$}}]{lemma}{CaseOne}\label{le:Case1}
% In Case 1, the constructed drawing $\Gamma$ is a $uv$-separated drawing of~$G[u,v]$ satisfying 
% \hyperref[prop:W]{Properties W} and \hyperref[prop:H]{H}.
% \end{restatable}

\begin{proof}
First, it trivially comes from the construction that $\Gamma$ is a straight-line grid drawing and that it preserves the given embedding, provided its planarity. 

We now prove that $\Gamma$ is planar. The paths $P_0,P_1,\dots,P_k$ do not cross each other in $\Gamma$, as they lie on distinct horizontal lines $h_0,h_1,\dots,h_k$. The other edges that are directly drawn by the algorithm (we call them \emph{binding edges}) connect vertices on a path $P_i$ and a path~$P_{i+1}$, hence they do no cross any of the paths $P_0,P_1,\dots,P_k$. Also, no two binding edges cross each other in $\Gamma$, given that the left-to-right order of the vertices of~$P_i$ and $P_{i+1}$ on~$h_i$ and~$h_{i+1}$ is the same as the order in which such vertices are encountered when traversing in counter-clockwise direction the boundary of the outer face of~$G$ from $u$ to $v$. Any two recursively drawn graphs $G^i_j$ do not cross each other in $\Gamma$, as they are horizontally disjoint, by construction. Likewise, for any $l>i$, any  recursively drawn graph $G^i_j$ does not cross any edge of~$P_l$ or any binding edge incident to a vertex of~$P_l$, as $G^i_j$ and such an edge are horizontally-disjoint in~$\Gamma$. Also, for any $l<i$, any   recursively drawn graph $G^i_j$ does not cross any edge of~$P_l$ or any binding edge incident to a vertex of~$P_l$, as the horizontal line $h_i$ keeps $G^i_j$ and such an edge on different sides in $\Gamma$; the fact that $G^i_j$ lies below $h_i$ is a consequence of \cref{prp:1,prp:2,prp:3} of a $v^i_jv^i_{j+1}$-separated drawing. Finally, $P_i$ and $G^i_j$ share the edge $(s^i_j,s^i_{j+1})$ and nothing else.

%Since the construction of~$\Gamma$ respects the embedding of~$G$, it follows that $\Gamma$ is outerplanar.

The convexity of each internal face $f$ of~$G$ in $\Gamma$ follows by induction if $f$ is internal to some recursively drawn graph $G^i_j$. Otherwise, $f$ is convex as it is drawn as a triangle or a trapezoid in $\Gamma$, since its incident vertices lie on two horizontal lines $h_i$ and $h_{i+1}$, for some $0\leq i\leq k-1$. \cref{prp:1,prp:2,prp:3,prp:4} of a $uv$-separated drawing are easily implied by our construction. In particular, $u$ and $v$ are the vertices last placed by our construction, and their placement satisfies \cref{prp:1} and \cref{prp:4}. The neighbors of~$u$ and $v$ all lie on $h_1$, which is one unit below the horizontal line $h_0$ through $u$ and $v$, thus satisfying \cref{prp:2}. Finally, each vertex of~$G$ different from $u$ and $v$ lies on $h_1,\dots,h_k$ or belongs to a recursively constructed drawing whose top side is on one of such lines. This ensures \cref{prp:3}.

In order to complete the proof, it remains to prove \hyperref[prop:W]{\em Property W} and \hyperref[prop:W]{\em Property H}.

\hyperref[prop:W]{\em Property W} easily follows by induction, given that every vertical grid line intersecting $\Gamma$ intersects a recursively constructed drawing or a vertex directly placed by the algorithm. 

The first part of \hyperref[prop:H]{\em Property H} also follows by induction, given that every horizontal grid line intersecting $\Gamma$ intersects a recursively constructed drawing or a vertex directly placed by the algorithm. Concerning the second part, we have that $h(\Gamma)$ is at most $k$ plus the maximum height of a recursively constructed drawing $\Gamma^i_j$. Indeed, the paths $P_0,P_1,\dots,P_k$ occupy $k+1$ consecutive horizontal lines $h_0,h_1,\dots,h_k$, respectively. No vertex of~$\Gamma$ lies above $h_0$. Furthermore, every recursively constructed drawing $\Gamma^i_j$ is placed in $\Gamma$ so that it intersects at least one of the lines $h_0,h_1,\dots,h_k$. Note that $k\leq \sqrt n$ by assumption and hence $k\leq \sqrt d$, given that $d\geq n$. Assume that $i$ and $j$ are the indices such that $h(\Gamma^i_j)$ is maximum. Let $d^i_j$ be the number of vertices in the extended weak dual tree of~$G^i_j$. Hence, we have $h(\Gamma)\leq \sqrt d+h(\Gamma^i_j)$ and thus, by induction, $h(\Gamma)\leq \sqrt d+f(d^i_j)$. Since $G^i_j$ is a subgraph of a left subgraph of~$G$ at $\pi$ or of a right subgraph of~$G$ at $\pi$, by \cref{th:decomposition}, it holds true that $(d^i_j)^p\leq (1-\delta)d^p$. By the definition of~$f(d)$ in \cref{le:induction_works} (with $d_a=d^i_j$ and $d_b=0$), we have that $\sqrt d+f(d^i_j)\leq f(d)$, hence $h(\Gamma)\leq f(d)$, which proves \hyperref[prop:H]{\em Property H}.
\end{proof}

% \hly{Gio: we need to argue that the height is in $O(\sqrt{n})$ in this case. I believe the recurrence is $h(n) \leq 1+ \sqrt{n} + h(n/2)$ which solves to $h(n) \in O(\sqrt{n})$ (I am counting the span, should we count the intercepted horizontal lines instead? in this case it would be $2+ \dots $). However, my impression is that this recurrence comes from drawing the spine first, which is not what we are doing now.} 

% \hly{Bri: It is not guaranteed that the subtrees have size at most n/2. Let's postpone this until we have the recurrence for the second construction as well. We should be able to use a lemma (or a version of it) from Beppe's paper.}

We now describe the construction of {\bf Case 2}. Recall that, in this case, we have $k\geq 1+\sqrt n$ We start the description with the following simple combinatorial lemma. 

\begin{lemma} \label{le:small-subgraph}
There exists an index $t$ with $2\leq t\leq 1+\sqrt n$ such that $|V_t|\leq \sqrt n$. 
\end{lemma}

\begin{proof}
Since the sets $V_2,V_3,\dots,V_{1+\sqrt n}$ are disjoint and since there are $\sqrt n$ of them, the smallest among them contains at most $\sqrt n$ vertices.   
\end{proof}

% (actually, strictly less than $\sqrt n$ vertices, given that $|V_0\cup V_1|\geq 4$)

The construction for Case 2 proceeds in three steps. In \textbf{Step 1}, we construct a drawing $\Gamma^*_t$ of~$G^*_t$; in \textbf{Step 2}, we augment $\Gamma^*_t$ to a drawing $\Gamma^*_{t-1}$ of~$G^*_{t-1}$; finally, in \textbf{Step 3}, we augment $\Gamma^*_{t-1}$ to the desired drawing $\Gamma$ of~$G$.

\begin{figure}[tb]
    \centering   \includegraphics[page=3,width=\textwidth]{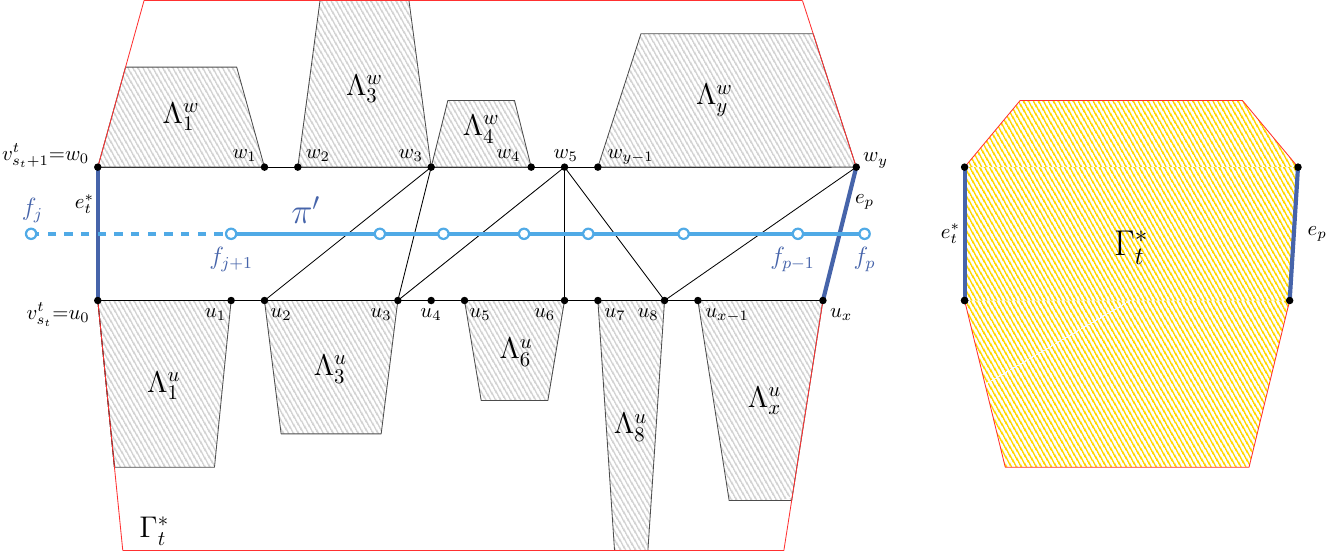}
    \caption{(Left) Construction of the drawing $\Gamma^*_t$ of~$G^*_t$. Labels $\Lambda^u_j$ and $\Lambda^v_j$ are omitted for single edges. (Right) Polygon used to represent the drawing of~$\Gamma^*_t$ in the subsequent figures.}
    \label{fig:drawing-G*_t}
\end{figure}

\subparagraph{\bf Step 1:} We construct a drawing $\Gamma^*_t$ of~$G^*_t$; refer to \cref{fig:drawing-G*_t}. Recall that the transition edge $e^*_t$ has end-vertices $v^t_{s_t}$ and $v^t_{s_t+1}$ and that $G^*_t$ is the graph $G[v^t_{s_t},v^t_{s_t+1}]$. Let $(f_j,f_{j+1})$ be the edge of~$\pi$ dual to $e^*_t$, let $\pi'$ be the subpath $(f_{j+1},f_{j+2},\dots,f_p)$ of~$\pi$, and let  $(u_0=v^t_{s_t},u_1,\dots,u_x,w_y,w_{y-1},\dots,w_1,w_0=v^t_{s_t+1})$ be the counter-clockwise order of the vertices along the outer face of~$G[\pi']$, where $(u_x,w_y)$ is the edge $e_p$. 
We initialize $\Gamma^*_t$ by drawing~$e^*_t$ as a vertical segment of height $1$. For $j=1,\dots,x$, we recursively construct a drawing $\Lambda^u_j$ of~$G[u_{j-1},u_j]$. We place the drawings $\Lambda^u_1,\dots,\Lambda^u_x$ side by side, so that the placement of a vertex $u_j$ coincides in the drawings $\Lambda^u_j$ and $\Lambda^u_{j+1}$ where it appears, for $j=1,\dots,x-1$, and so that the placement of~$u_0$ coincides in $\Lambda^u_1$ and in the drawing of~$e^*_t$. Analogously, for $j=1,\dots,y$, we recursively construct a drawing $\Lambda^w_j$ of~$G[w_j,w_{j-1}]$. Unlike for $\Lambda^u_1,\dots,\Lambda^u_x$, we rotate the drawings $\Lambda^w_1,\dots,\Lambda^w_y$ by $180^\circ$, and place them side by side, so that the placement of a vertex $w_j$ coincides in $\Lambda^w_j$ and $\Lambda^w_{j+1}$, for $j=1,\dots,y-1$, and so that the placement of~$w_0$ coincides in $\Lambda^w_1$ and in the drawing of~$e^*_t$. This completes the construction of~$\Gamma^*_t$.

\begin{figure}[tb!]
    \centering
    \includegraphics[page=4, width=\textwidth]{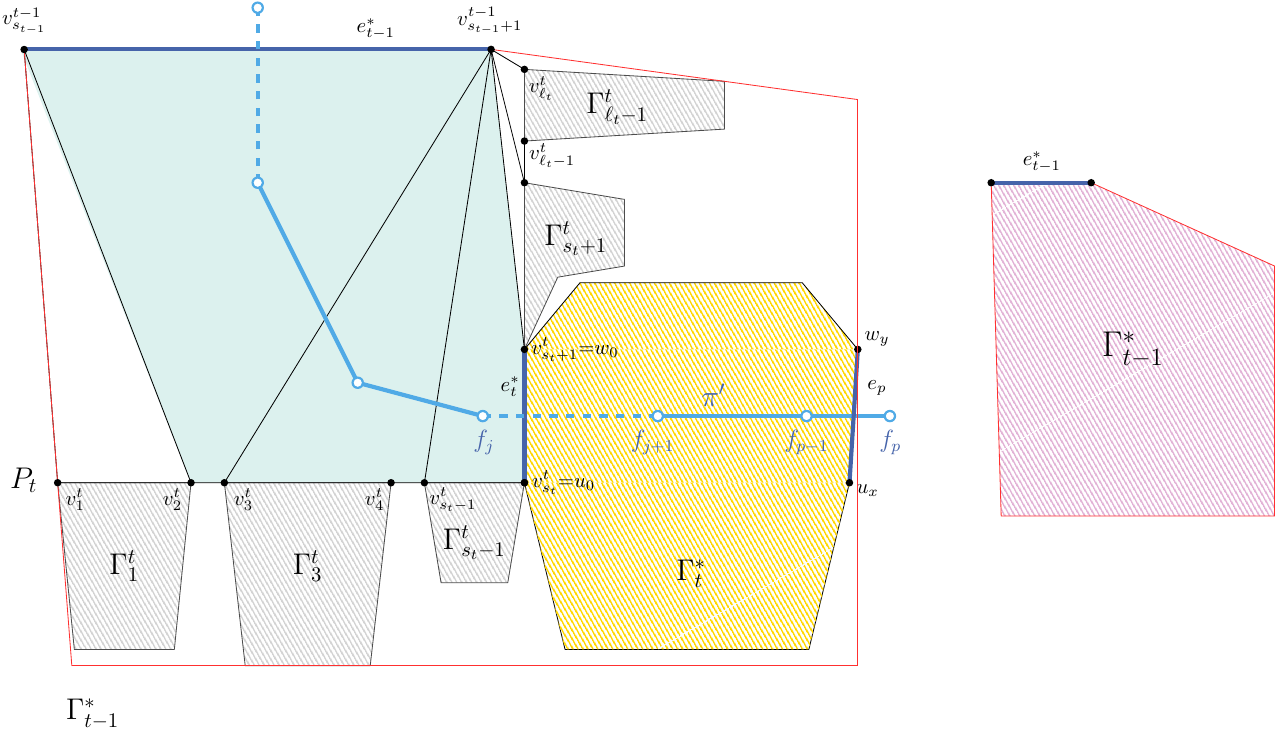}\hfill
    \caption{(Left) Construction of $\Gamma^*_{t-1}$ from $\Gamma^*_t$, whose boundary is depicted by the yellow-tiled region. (Right) Polygon used to represent~$\Gamma^*_{t-1}$ in the subsequent figures.}
    \label{fig:drawing-G*_t-1}
\end{figure}

\subparagraph{\bf Step 2:} 
We now augment $\Gamma^*_t$ to a drawing $\Gamma^*_{t-1}$ of~$G^*_{t-1}$; refer to \cref{fig:drawing-G*_t-1}. Recall that the path $P_t$ has vertices $(v^t_1,v^t_2,\dots,v^t_{s_t},v^t_{s_t+1},\dots,v^t_{\ell_t})$. For $j=1,\dots,\ell_t-1$ with $j\neq s_t$, we recursively construct a $v^t_jv^t_{j+1}$-separated drawing $\Gamma^t_j$ of~$G[v^t_j,v^t_{j+1}]$. 
We place the recursively constructed drawings $\Gamma^t_1,\dots,\Gamma^t_{s_t-1}$ side by side, so that the placement of a vertex $v^t_j$ coincides in $\Gamma^t_{j-1}$ and $\Gamma^t_{j}$, for $j=2,\dots,s_t-1$, and so that the placement of~$v^t_{s_t}$ coincides in $\Gamma^t_{s_t-1}$ and in $\Gamma^*_t$. Next, we rotate the drawings $\Gamma^t_{s_t+1},\dots,\Gamma^t_{\ell_t-1}$ in counter-clockwise direction by $90^\circ$, and we stack them one on top of the other, so that the placement of a vertex $v^t_j$ coincides in $\Gamma^t_{j-1}$ and $\Gamma^t_{j}$, for $j=s_t+2,\dots,\ell_t-1$, and so that the placement of~$v^t_{s_t+1}$ in  $\Gamma^t_{s_t+1}$ is on the same vertical line as in $\Gamma^*_t$ and on the same horizontal line as the highest vertex in $\Gamma^*_t$ (refer to the left side of \cref{fig:change}). 
\begin{figure}[tb!]
  \begin{center}
\includegraphics[page=5,width=.5\textwidth]{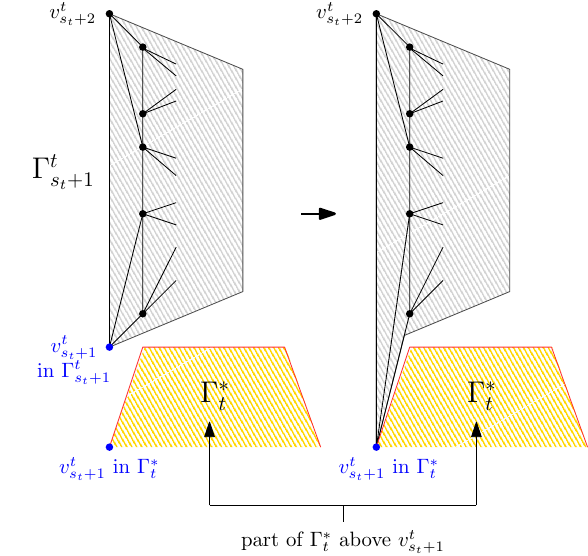}
  \end{center}
  \caption{
  Transformation of $\Gamma^t_{s_t+1}$ shifting downward $v^t_{s_t+1}$ while avoiding overlapping with $\Gamma^*_t$.
  }\label{fig:change}
\end{figure}
This might provide a double placement for $v^t_{s_t+1}$ (unless $v^t_{s_t+1}$ is the highest vertex in $\Gamma^*_t$). The issue is resolved by keeping for $v^t_{s_t+1}$ the position it has in $\Gamma^*_t$; note that, in $\Gamma^t_{s_t+1}$, this amounts to shifting $v^t_{s_t+1}$ downwards (refer to the right side of \cref{fig:change}). The drawing of~$\Gamma^*_{t-1}$ is completed by placing $v^{t-1}_{s_{t-1}}$ and $v^{t-1}_{s_{t-1}+1}$ on the same horizontal line, one unit higher than the highest vertex in $\Gamma^*_{t-1}$ so far (this is either $v^t_{\ell_t}$ if $s_t\leq \ell_t-2$ or the highest vertex in $\Gamma^*_t$ if $s_t=\ell_t-1$), so that $v^{t-1}_{s_{t-1}+1}$ is one unit to the left of the vertical line through $e^*_t$ and $v^{t-1}_{s_{t-1}}$ is one unit to the left of the leftmost vertex in $\Gamma^*_{t-1}$ so far (this is either $v^t_1$ if $s_t\geq 2$ or $v^{t-1}_{s_{t-1}+1}$ if $s_t=1$).

\begin{figure}[tb!]
\centering
\includegraphics[page=7, width=.9\textwidth]{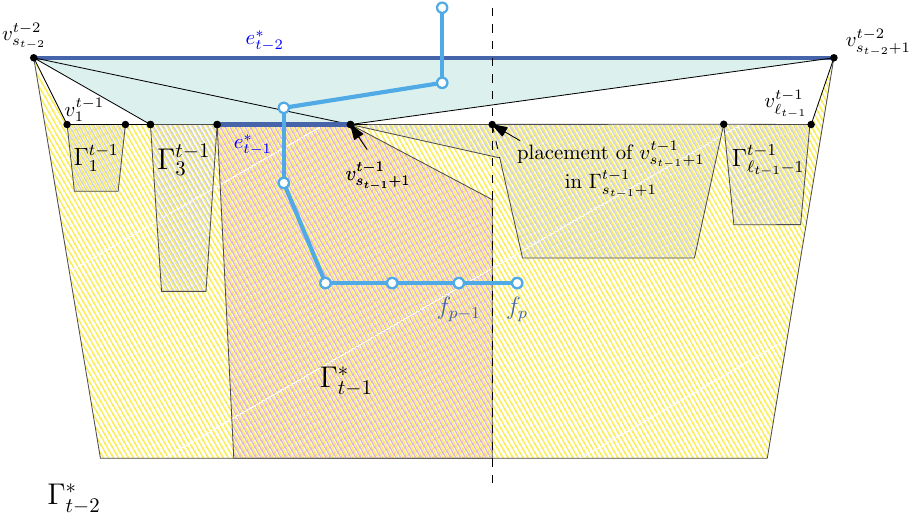}\hfill
\caption{Construction of $\Gamma^*_{t-2}$ from $\Gamma^*_{t-1}$, whose boundary is depicted by the red-tiled region. The illustration shows the modification of the placement of~$v^{t-1}_{s_{t-1} +1}$ from the one it has in $\Gamma^{t-1}_{s_{t-1} +1}$.}
\label{fig:drawing-G*_t-2}
\end{figure}
\subparagraph{\bf Step 3:} Finally, we show how to augment $\Gamma^*_{t-1}$ to a $uv$-separated drawing of~$G[u,v]$. This is done similarly to Case 1. Namely, we construct, for $i=t-2,t-3,\dots,0$, a  $v^i_{s_i}v^i_{s_i+1}$-separated drawing $\Gamma^*_i$ of 
$G[v^i_{s_i},v^i_{s_i+1}]$ from an already constructed  $v^{i+1}_{s_{i+1}}v^{i+1}_{s_{i+1}+1}$-separated drawing $\Gamma^*_{i+1}$ of  $G[v^{i+1}_{s_{i+1}},v^{i+1}_{s_{i+1}+1}]$. There is one exception to this strategy. Namely, in Step 2, we constructed $\Gamma^*_{t-1}$, which is not a $v^{t-1}_{s_{t-1}}v^{t-1}_{s_{t-1}+1}$-separated drawing of  $G[v^{t-1}_{s_{t-1}},v^{t-1}_{s_{t-1}+1}]$, as it does not satisfy \cref{prp:2,prp:4}. However, $\Gamma^*_{t-1}$ does satisfy \cref{prp:1,prp:3} and this is enough for the construction of a $v^{t-2}_{s_{t-2}}v^{t-2}_{s_{t-2}+1}$-separated drawing $\Gamma^*_{t-2}$ of  $G[v^{t-2}_{s_{t-2}},v^{t-2}_{s_{t-2}+1}]$. As in Case 1, we place the recursively constructed drawings $\Gamma^{i+1}_j$ with $j=1,\dots,\ell_{i+1}-1$ with $j\neq s_{i+1}$, as well as the drawing $\Gamma^*_{i+1}$, side by side, so that the placement of a vertex $v^{i+1}_j$ coincides in two drawings it belongs to, with one exception in the case $i=t-2$; refer to \cref{fig:drawing-G*_t-2}. Namely, the drawing $\Gamma^{t-1}_{s_{t-1}+1}$ of~$G^{t-1}_{s_{t-1}+1}$ is embedded so that $v^{t-1}_{s_{t-1}+1}$ is on the same horizontal line as in $\Gamma^*_{t-1}$ and on the same vertical line as the rightmost vertex in $\Gamma^*_{t-1}$. This provides a double placement for $v^{t-1}_{s_{t-1}+1}$. The issue is resolved by keeping for $v^{t-1}_{s_{t-1}+1}$ the position it has in $\Gamma^*_{t-1}$; note that, in $\Gamma^{t-1}_{s_{t-1}+1}$, this amounts to shifting $v^{t-1}_{s_{t-1}+1}$ leftwards. The desired $v^i_{s_{i}}v^i_{s_{i}+1}$-separated drawing $\Gamma^*_i$ of 
$G[v^i_{s_i},v^i_{s_i+1}]$ is completed by placing $v^i_{s_{i}}$ one unit above and one unit to the left of the leftmost vertex of~$P_{i+1}$, and $v^i_{s_{i}+1}$ one unit above and one unit to the right of the rightmost vertex of~$P_{i+1}$. When $i=0$, the obtained drawing $\Gamma:=\Gamma^*_0$ is a $uv$-separated drawing of~$G[u,v]$, as proved in the following.

\begin{lemma}\label{le:Case2}
In Case 2, the constructed drawing $\Gamma$ is a $uv$-separated drawing of~$G[u,v]$ satisfying \hyperref[prop:W]{\em Property W} and \hyperref[prop:H]{\em H}.
\end{lemma}

\begin{proof}
First, it trivially comes from the construction that $\Gamma$ is a straight-line grid drawing and that it preserves the given embedding, provided its planarity.

We now prove that $\Gamma$ is planar. In order to do that, we prove that the drawing produced at each step of the algorithm's construction is planar. As in the proof of \cref{le:Case1}, our arguments use (without explicitly mentioning it) the fact that every recursively constructed $z_1z_2$-separated drawing lies, apart from the edge $(z_1,z_2)$, below the horizontal line through $z_1$ and $z_2$. This is a consequence of \cref{prp:1,prp:2,prp:3} of a $z_1z_2$-separated drawing.
\begin{itemize}
\item {\em Step 1.} The drawing $\Gamma^*_t$ is easily proved to be planar. Indeed, the paths $(v^t_{s_t}=u_0,u_1,\dots,u_x)$ and $(v^t_{s_t+1}=w_0,w_1,\dots,w_y)$ on the boundary of~$G[\pi']$ lie on two horizontal lines in $\Gamma^*_t$. This implies that they do not cross each other, that they do not cross any edge $(u_i,w_j)$, and that no two edges $(u_i,w_j)$ and $(u_{i'},w_{j'})$ cross each other. Each graph $G[u_i,u_{i+1}]$ does not cross any graph $G[u_j,u_{j+1}]$, as such graphs are horizontally disjoint in $\Gamma^*_t$, does not cross any subgraph $G[w_j,w_{j+1}]$ as it is vertically disjoint from it, and does not cross $G[\pi']$ as it is vertically disjoint from it, except for the edge $(u_i,u_{i+1})$ which is shared by the two graphs. Similarly, each graph $G[w_j,w_{j+1}]$ is free of crossings in $\Gamma^*_t$.
\item {\em Step 2.} We now deal with $\Gamma^*_{t-1}$. Path $P_t$ is represented in $\Gamma^*_{t-1}$ as a polygonal line composed of a vertical segment on a line $\ell$ and of a horizontal segment to the left of~$\ell$. Since $G^*_{t}$ lies entirely to the right of  $\ell$ in $\Gamma^*_{t-1}$, except for the edge $(v^t_{s_t},v^t_{s_t+1})$ which is shared with $P_t$, it follows that $P_t$ does not cross $G^*_{t}$. The recursively drawn graphs $G^t_j$ with $j<s_t$ do not cross each other as they are horizontally disjoint, and do not cross~$G^*_{t}$ and the recursively drawn graphs~$G^t_j$ with~$j>s_t$ as they are horizontally disjoint from them in $\Gamma^*_{t-1}$. Also, the recursively drawn graphs $G^t_j$ with $j>s_t$ do not cross each other as they are vertically disjoint, and the recursively drawn graphs $G^t_j$ with $j>s_t+1$ do not cross $G^*_t$ as they are vertically disjoint from it. The proof that $G^t_{s_t+1}$ and $G^*_t$ do not cross each other in $\Gamma^*_{t-1}$ uses the properties of a $v^t_{s_t+1}v^t_{s_t+2}$-separated drawing (and in fact, the avoidance of such crossings is the driving force behind the definition of our drawing invariant). Indeed, by \cref{prp:4} and by the $90^\circ$ counter-clockwise rotation of~$\Gamma^t_{s_t+1}$, we have that $v^t_{s_t+1}$ is the lowest vertex in $\Gamma^t_{s_t+1}$. Since $\Gamma^t_{s_t+1}$ is placed so that $v^t_{s_t+1}$ is on the same horizontal line as the highest vertex in $\Gamma^*_t$, it follows that the only edges of~$G^t_{s_t+1}$ that might cross $G^*_t$ are those incident to $v^t_{s_t+1}$. However, since all the vertices of~$G^t_{s_t+1}$ different from $v^t_{s_t+1}$ lie above $\Gamma^*_t$ and since all the neighbors of~$v^t_{s_t+1}$ in $G^t_{s_t+1}$ lie on the vertical line one unit to the right of~$\ell$, by \cref{prp:2} and by the $90^\circ$ counter-clockwise rotation of~$\Gamma^t_{s_t+1}$, it follows that the edges of~$G^t_{s_t+1}$ incident to $v^t_{s_t+1}$ lie above every edge of~$G^*_t$ with which they have a horizontal overlap. We conclude the proof of planarity of~$\Gamma^*_{t-1}$ by remarking that, by \cref{pro:shift}, placing $v^t_{s_t+1}$ at the point where it is placed in $\Gamma^*_t$ maintains the planarity of~$\Gamma^t_{s_t+1}$. 
\item {\em Step 3.} That the paths $P_0,P_1,\dots,P_{t-1}$, the binding edges connecting vertices on such paths, and the recursively drawn graphs $G^i_j$ (as well as the graph $G^*_{t-1}$) sharing an edge with these paths do not cross each other in $\Gamma$ is proved by the same arguments as in Case~1, with one exception. Namely, $G^*_{t-1}$ is not horizontally disjoint in $\Gamma$ from $G^{t-1}_{s_{t-1}+1}$. To prove that these two graphs do not cross each other in $\Gamma$, we use the same argument as the one employed for Step 2.   By \cref{prp:4} and since $\Gamma^{t-1}_{s_{t-1}+1}$ is placed so that $v^{t-1}_{s_{t-1}+1}$ is on the same vertical line as the rightmost vertex in $\Gamma^*_{t-1}$, the only edges of~$G^{t-1}_{s_{t-1}+1}$ that might cross $G^*_{t-1}$ are those incident to $v^{t-1}_{s_{t-1}+1}$. However, since all the vertices of~$G^{t-1}_{s_{t-1}+1}$ different from $v^{t-1}_{s_{t-1}+1}$ lie to the right of~$\Gamma^*_{t-1}$ and since all the neighbors of~$v^{t-1}_{s_{t-1}+1}$ in $G^{t-1}_{s_{t-1}+1}$ lie on the horizontal line one unit below $h_{t-1}$, by \cref{prp:2}, it follows that the edges of~$G^{t-1}_{s_{t-1}+1}$ incident to $v^{t-1}_{s_{t-1}+1}$ lie to the right of every edge of~$G^*_{t-1}$ with which they have a vertical overlap. Finally, note that, by \cref{pro:shift}, placing $v^{t-1}_{s_{t-1}+1}$ at the point where it is placed in $\Gamma^*_{t-1}$ maintains the planarity of~$\Gamma^{t-1}_{s_{t-1}+1}$.
\end{itemize}
%Since the construction of~$\Gamma$ respects the embedding of~$G$, it follows that $\Gamma$ is outerplanar.

The convexity of each internal face $f$ of~$G$ in $\Gamma$ follows by induction if $f$ is internal to some recursively drawn subgraph $G^i_j$ of~$G$. If $f$ is an internal face of~$G[\pi']$ or if $f$ has vertices on two distinct paths among $P_0,P_1,\dots,P_{t-1}$, then $f$ is drawn as a triangle or a trapezoid in~$\Gamma$, hence it is convex. The only faces that do not fit in any of these two categories are those incident to $v^{t-1}_{s_{t-1}}$ and/or $v^{t-1}_{s_{t-1}+1}$ and incident to at least one vertex in $P_t$.  The convexity of these faces follows from two facts. 

\begin{itemize}
\item First, the polygon $Q$ delimited by the edge $(v^{t-1}_{s_{t-1}},v^{t-1}_{s_{t-1}+1})$ and by $P_t$ is a convex pentagon $(v^{t-1}_{s_{t-1}},v^{t-1}_{s_{t-1}+1},v^t_{\ell_t},v^t_{s_t},v^t_{1})$ (if $s_t>1$) or a convex quadrilateral $(v^{t-1}_{s_{t-1}},v^{t-1}_{s_{t-1}+1},v^t_{\ell_t},v^t_{1})$ (if $s_t=1$). This is a direct consequence of the construction, which ensures that $v^{t-1}_{s_{t-1}}$ and $v^{t-1}_{s_{t-1}+1}$ lie on a horizontal line higher than $v^t_{\ell_t}$, $v^t_{s_t}$, and $v^t_{1}$, with $v^{t-1}_{s_{t-1}+1}$ to the right of~$v^{t-1}_{s_{t-1}}$, that $v^t_{\ell_t}$ and $v^t_{s_t}$ lie on a vertical line to the right of~$v^{t-1}_{s_{t-1}}$ and $v^{t-1}_{s_{t-1}+1}$ (and $v^t_{1}$ if $s_t>1$), with $v^t_{\ell_t}$ higher than $v^t_{s_t}$, and that, if $s_t>1$, $v^t_{1}$ and $v^t_{s_t}$ lie on a horizontal line lower than $v^t_{\ell_t}$, $v^{t-1}_{s_{t-1}+1}$, and $v^{t-1}_{s_{t-1}}$, with $v^t_{s_t}$ to the right of~$v^t_{1}$.
\item Second, the faces incident to $v^{t-1}_{s_{t-1}}$ and/or $v^{t-1}_{s_{t-1}+1}$ and incident to at least one vertex in~$P_t$ form a convex subdivision of~$Q$.
\end{itemize}

The fact that $\Gamma$ satisfies \cref{prp:1,prp:2,prp:3,prp:4} of a $uv$-separated drawing can be proved as in \cref{le:Case1}, given that the ``topmost'' part of~$\Gamma$ is constructed as in Case~1, with the only addendum that $\Gamma^*_{t-1}$ lies entirely below or on $h_{t-1}$, by construction; this is needed to ensure \cref{prp:3} for $\Gamma$. Note that the assumption $t\geq 2$ guarantees that Step 3 is actually applied for the construction of $\Gamma$. This is vital since, as mentioned before, $\Gamma^*_{t-1}$ is not a $v^{t-1}_{s_{t-1}}v^{t-1}_{s_{t-1}+1}$-separated drawing of $G[v^{t-1}_{s_{t-1}},v^{t-1}_{s_{t-1}+1}]$, and hence, if $t=1$ were allowed, the drawing $\Gamma=\Gamma^*_0$ might not be a $uv$-separated drawing of $G$. 

Finally, we prove \hyperref[prop:W]{\em Property W} and \hyperref[prop:H]{\em Property H}. Actually, \hyperref[prop:W]{\em Property W} and the first part of \hyperref[prop:H]{\em Property~H} are proved as in the proof of \cref{le:Case1}. Concerning the second part of \hyperref[prop:H]{\em Property~H}, we start by proving that $h(\Gamma)$ is at most $t$, plus $\sqrt n$, plus the maximum height of a recursively constructed drawing $\Gamma^i_j$ of a subgraph $G^i_j$ of a left subgraph of~$G$ at $\pi$, plus the maximum height of a recursively constructed drawing $\Gamma^{i'}_{j'}$ of a subgraph $G^{i'}_{j'}$ of a right subgraph of~$G$ at~$\pi$. First, no vertex of~$\Gamma$ lies above $h_0$; this fact, which is trivial in the proof of \cref{le:Case1}, uses here the observation that $\Gamma^*_{t-1}$ lies entirely below or on $h_{t-1}$, by construction. The horizontal grid lines $h_0,h_1,\dots,h_{t-1}$ on which $P_0,P_1,\dots,P_{t-1}$ are placed give rise to the~$t$ term in the upper bound for $h(\Gamma)$. The $\sqrt n$ term accounts for the horizontal grid lines between the highest line intersecting~$\Gamma^t_{\ell_t-1}$ (which is one unit below $h_{t-1}$) and the lowest line intersecting $\Gamma^t_{s_t+1}$ before $v^t_{s_t+1}$ is moved to the position it has in $\Gamma^*_t$. Indeed, since the drawings $\Gamma^t_{s_t+1},\Gamma^t_{s_t+2},\dots,\Gamma^t_{\ell_t-1}$ are constructed recursively, and hence satisfy \hyperref[prop:W]{\em Property~W}, and are rotated by $90^\circ$, and since the graphs $G^t_{s_t+1},G^t_{s_t+2},\dots,G^t_{\ell_t-1}$ have a total of at most $\sqrt n$ vertices, by \cref{le:small-subgraph}, it follows that the drawings $\Gamma^t_{s_t+1},\Gamma^t_{s_t+2},\dots,\Gamma^t_{\ell_t-1}$ intersect, before moving $v^t_{s_t+1}$, at most $\sqrt n$ grid lines. 
%By \cref{prp:4}, all the vertices of~$G^t_{s_t+1},G^t_{s_t+2},\dots,G^t_{\ell_t-1}$ lie between the horizontal lines through $v^t_{s_t+1}$ and $v^t_{\ell_t}$ in $\Gamma$, hence the number of grid lines between the highest line intersecting $\Gamma^t_{\ell_t-1}$ and the lowest line intersecting $\Gamma^t_{s_t+1}$ is at most $\sqrt n$, as desired.
Next, we observe that $h(\Gamma^*_{t})$ is given by the sum of the maximum height of a graph $G[u_j,u_{j+1}]$ and the maximum height of a graph $G[w_j,w_{j+1}]$. Since the former are left subgraphs of~$G$ at $\pi$ and the latter are right subgraphs of~$G$ at $\pi$, this gives rise to the term maximum height of a recursively constructed drawing $\Gamma^i_j$ of a subgraph $G^i_j$ of a left subgraph of~$G$ at $\pi$ and to the term maximum height of a recursively constructed drawing $\Gamma^{i'}_{j'}$ of a subgraph $G^{i'}_{j'}$ of a right subgraph of~$G$ at $\pi$ in the upper bound for $h(\Gamma)$. The proof of the claimed upper bound on $h(\Gamma)$ is concluded by considering the recursively drawn graphs that were not dealt with yet. Every recursively constructed drawing~$\Gamma^i_j$ with $i\leq t-1$ is placed in $\Gamma$ so that it intersects at least one of the lines $h_0,h_1,\dots,h_{t-1}$, hence the number of horizontal grid lines between  $h_0$ and the lowest line intersecting one of such drawings~$\Gamma^i_j$ is at most $t-1+ h(\Gamma^i_j)$, which is smaller than the claimed upper bound. Finally, the recursively drawn graphs $G^t_j$ with $j< s_t$ have their top side, defined by the edge $(v^t_j,v^t_{j+1})$, on the horizontal line through vertices $u_0,u_1,\dots,u_x$. Since these graphs are left subgraphs of~$G$ at~$\pi$, same as the graphs $G[u_j,u_{j+1}]$, it follows that the number of horizontal grid lines between  $h_0$ and the lowest line intersecting $G^t_j$ with $j<s_t$ is at most the claimed upper bound on $h(\Gamma)$, given that the same is true for $G[u_j,u_{j+1}]$.

%
%
% occupy $t$ consecutive horizontal lines , respectively. Nis here a consequence of the placement of~$v^{t-1}_{s_{t-1}}$ and $v^{t-1}_{s_{t-1}+1}$ as the highest vertices of~$\Gamma^*_{t-1}$, which ensures that $\Gamma^*_{t-1}$ lies entirely on or below $h_{t-1}$. Consider now the lowest vertex $v_l$ in $\Gamma$. Since no vertex of~$\Gamma$ lies above $h_0$, in order to prove the claimed bound on $h(\Gamma)$ we need to prove that $v_l$ argue about the distance between $h_0$ and $v_l$. We distinguish some cases.
%
%\begin{itemize}
%\item If $v_l$ is not in $G^*_{t-1}$ or if it coincides with $v^{t-1}_{s_{t-1}}$ or with $v^{t-1}_{s_{t-1}+1}$, then $v_l$ belongs to one of the paths $P_0,P_1,\dots,P_{t-1}$, or to a graph $G^i_j$ with $j\leq t-1$. Then as in the proof of \cref{le:Case1}, it suffices to observe that every recursively constructed drawing $\Gamma^i_j$ with $j\leq t-1$ is placed in $\Gamma$ so that it intersects at least one of the lines $h_0,h_1,\dots,h_{t-1}$. Indeed, this implies that $h(\Gamma)\leq t-1 + \max h(G^i_j)$, where the maximum is over all recursively drawn graphs $G^i_j$ with $i\leq t-1$, and hence that $h(\Gamma)$ is at most $t-1$ plus the maximum height of a recursively constructed drawing $\Gamma^i_j$ of a (proper or not proper) left or right subgraph $G^i_j$ of~$G$ at $\pi$, as desired.
%
%\todo[inline]{Finish these cases}
%\end{itemize}

Note that $t\leq \sqrt n$ by construction and hence $t\leq \sqrt d$, given that $d\geq n$. Assume that $i$ and $j$ are indices such that $G^i_j$ is a recursively drawn subgraph of a left subgraph of~$G$ at $\pi$ and $h(\Gamma^i_j)$ is maximum. Also, assume that $i'$ and $j'$ are indices such that $G^{i'}_{j'}$ is a  recursively drawn subgraph of a right subgraph of~$G$ at $\pi$ and $h(\Gamma^{i'}_{j'})$ is maximum. Let $d^i_j$ and $d^{i'}_{j'}$ be the number of vertices in the extended weak dual trees of~$G^i_j$ and $G^{i'}_{j'}$, respectively. By the upper bound proved above, we have $h(\Gamma)\leq 2\sqrt d+h(\Gamma^i_j)+h(\Gamma^{i'}_{j'})$ and thus, by induction, $h(\Gamma)\leq 2\sqrt d+f(d^i_j)+f(d^{i'}_{j'})$. By \cref{th:decomposition}, it holds true that $(d^i_j)^p+(d^{i'}_{j'})^p\leq (1-\delta)d^p$. By the definition of~$f(d)$ in \cref{le:induction_works} (with $d_a=d^i_j$ and $d_b=d^{i'}_{j'}$), we have that $2\sqrt d+f(d^i_j)+f(d^{i'}_{j'})\leq f(d)$, hence $h(\Gamma)\leq f(d)$, which proves \hyperref[prop:H]{\em Property H}.
\end{proof}

\section{Internally-Strictly-Convex Drawings} \label{se:strict}

In this section, we present two algorithms for obtaining small-area internally-strictly-convex drawings of outerpaths of bounded face size (\cref{sse:outerpaths}) and of outerplanar graphs of bounded diameter (\cref{sse:diameter}). Notably, these results imply that meaningful families of outer-$1$-plane graphs admit an embedding-preserving straight-line drawing on an integer grid of polynomial size.

\subsection{Outerpaths of Bounded Face Size}
\label{sse:outerpaths}

In this subsection, we prove a tight bound on the area required by internally-strictly-convex grid drawings of outerpaths. More precisely, we prove that every $n$-vertex outerpath whose internal faces have size at most $k$ admits an internally-strictly-convex grid drawing in area~$O(nk^2)$, which we also show to match a general lower bound for the area requirements of internally-strictly-convex grid drawings of outerplanar graphs. We start by proving the lower bound.

\begin{theorem} \label{th:lower}
Every $n$-vertex outerplane graph with $\Omega(\frac{n}{k})$ internal faces of size $k$ requires $\Omega(nk^2)$ area in any internally-strictly-convex grid drawing.
\end{theorem}

\begin{proof}
Let $G$ be an $n$-vertex outerplane graph with $\Omega(\frac{n}{k})$ internal faces of size $k$.  First, note that any internally-strictly-convex grid drawing of~$G$ must be outerplanar. Indeed, any embedding of an outerplanar graph which is not an outerplanar embedding contains a degree-$2$ vertex as an internal vertex. However, such a vertex would create an angle larger than or equal to $180^\circ$ in an internal face.
The area lower bound is obtained by a simple packing argument. Since any internally-strictly-convex grid drawing of~$G$ is outerplanar, it contains $\Omega(\frac{n}{k})$ internal faces of size $k$. Each of these faces occupies $\Omega(k^3)$ area \cite{Andrews1963}, hence the total area of the drawing is $\Omega(\frac{n}{k} \cdot k^3)$, which gives the bound claimed by the theorem.
\end{proof}

\begin{figure}[htb]
    \centering
    \includegraphics[page=12,width=\textwidth]{Figures/outerpath.pdf}
    \caption{Decomposition of an outerpath $G$ in fan graphs $G_0,G_1,\dots,G_r$. Gate edges are blue.}
    \label{fig:fan-decomposition}
\end{figure}
%
% \todo[inline]{l.361-386: Compared to the previous section where the superscripts and subscripts where intuitive, I feel that better choices could be made here, involving less obscure characters.\\

% Gio: I agree with the reviewer. Not for the next paragraph, but for the subsequent one describing the decomposition of $G$, where we use $\hat{e},\tilde{e}$ (also in Fig. 8) and the edges $e^a_i$ and $e^b_i$ (not in Fig. 8). I did not change anything yet. What would be better choices?
% }
%
In the rest of the section, we prove a matching upper bound for the case of outerpaths. Note that there exist $n$-vertex outerpaths with $\Omega(\frac{n}{k})$ internal faces of size $k$, hence the lower bound of \cref{th:lower} applies to outerpaths as well. We introduce some notation and definitions. Let $G$ be an outerpath and let $\pi=(f_1,f_2,\dots,f_{p-1})$ be the weak-dual of~$G$, which is a path; refer to \cref{fig:fan-decomposition}. Let $\hat{e}$ be the edge of~$G$ dual to the edge $(f_1,f_2)$ of~$\pi$ and let~$\tilde{e}$ be the edge of~$G$ dual to the edge $(f_{p-2},f_{p-1})$ of~$\pi$. Moreover, let $\hat{g} = (u_0,v_0)$ be any of the two external edges adjacent to $\hat{e}$ that belong to the face $f_1$ of~$G$, where $u_0$ immediately follows $v_0$ when traversing the boundary of the outer face of $G$ in counter-clockwise order. Also, let $\tilde{g}$ be any of the two external edges adjacent to $\tilde{e}$ that belong to the face $f_{p-1}$ of~$G$. 
We augment $\pi$ with two new vertices $f_0$ and $f_p$, and edges $(f_0,f_1)$ and $(f_{p-1},f_p)$.
We interpret $(f_0,f_1)$ and $(f_{p-1},f_p)$ as edges dual to $\hat{g}$ and $\tilde{g}$, respectively.

We now describe a decomposition of~$G$ into a sequence of smaller outerpaths such that any two consecutive outerpaths in the sequence share exactly one internal edge of~$G$.
We let~$e^*_0$ be the edge $\hat{g}$. Let $G_0$ be the plane graph induced by the vertices delimiting the internal faces of~$G$ incident to the end-vertex common to $e^*_0$ and $\hat{e}$ (see the vertex $u_0$ in \cref{fig:fan-decomposition}). 
Let~$e^*_1$ be the external edge of~$G_0$, other than $e^*_0$, dual to an edge of~$\pi$. We call $G_0$ a \emph{fan graph} with \emph{gate edges} $e^*_0$ and~$e^*_1$. Suppose that, for some $i\geq 1$, a fan graph $G_{i-1}$ with gate edges~$e^*_{i-1}$ and $e^*_i$ has been defined. Let $u_i$ and $v_i$ be the end-vertices of $e^*_i$, where $u_i$ is encountered before $v_i$ when traversing the boundary of the outer face of~$G$ in counter-clockwise direction from~$u_0$ to $v_0$.
We define the fan graph~$G_i$ as follows.
%Note that exactly one external edge $e^*_{i}$ of~$G_i$, other than $e^*_{i-1}$, is dual to an edge of~$\pi$.
%This edge $e^*_{i}$ is a \emph{gate edge}. 
% Let $e^a_i$ be the edge of~$\pi$ dual to $e^*_{i}$ and let $e^b_i$ be edge following $e^a_i$ when traversing $\pi$ from $f_0$ to $f_p$. Also, let $e^\diamond_i$ be the edge of~$G$ dual to $e^b_i$. 
% %
Let~$V_{i-1}$ be defined as $\bigcup^{i-1}_{j=0}V(G_j) \setminus \{u_i,v_i\}$. Then~$G_i$ is the outerplane graph induced by the vertices that do not belong to $V_{i-1}$ and that are incident to faces that have $u_i$ or $v_i$ on their boundary. 
Let $e^*_{i+1}$ be the external edge of~$G_i$, other than~$e^*_i$, dual to an edge of~$\pi$. The edges $e^*_{i}$ and $e^*_{i+1}$ are the gate edges of $G_i$.
Eventually, the decomposition is concluded by defining a fan graph $G_r$ with gate edges $e^*_r$ and $e^*_{r+1}$ such that $e^*_{r+1}=\tilde{g}$. If $\tilde{g}$ and $e^*_r$ are adjacent, this requires adding one vertex, two edges, and one internal face to $G$, so to ensure that $e^*_r$ and $e^*_{r+1}$ do not share any vertex.

% If $e^*_i$ and $e^\diamond_i$ do not share any end-vertex, then $G_i$ is the plane graph induced by the vertices delimiting the internal face of~$G$ incident to $e^*_i$ and $e^\diamond_i$ (see the face $f_5$ in \cref{fig:fan-decomposition} with $e^*_i = e^*_1$ and $e^\diamond_i = e^*_2$). In this case, $e^*_i$ and $e^*_{i+1}=e^\diamond_i$ are the gate edges of $G_i$.
% If, instead, $e^*_i=(u_i,v_i)$ and $e^\diamond_i$ share an end-vertex, say $u_i$, let us first define $V_{i-1} = \bigcup^{i-1}_{i=0}V(G_i)$.
% Then $G_i$ is the outerplane graph induced by the vertices delimiting the internal faces of~$G$ incident to $u_i$ that do not belong to $V_{i-1} \setminus \{u_i,v_i\}$. 

% Note that, if $r$ is the index such that $e^*_\ell$ belongs to $G_r$, then $e^*_\ell = e^*_{r+1}$ is a gate edge of~$G_r$\todo{Where does $e_{\ell}$ suddenly come from? Do you mean $e_{\ell}^*$? Gio: The reviewer is right. I changed $e_\ell$ to $e^*_\ell$. }.

%
We will show that each fan graph $G_i$ with gate edges $e^*_i$ and $e^*_{i+1}$ admits a \emph{block-drawing}. This is an outerplanar internally-strictly-convex grid drawing $\Gamma_i$ of~$G_i$ satisfying the following properties:

\begin{enumerate}[label=\bf B.\arabic*,ref=B.\arabic*,leftmargin=22pt]
    \item\label{prp:gate-edges} The two gate edges $e^*_i$ and $e^*_{i+1}$ are drawn as vertical segments of unit length such that $y(u_i)=y(u_{i+1})=0$ and $y(v_i)=y(v_{i+1})=1$;
    \item\label{prp:gate-left-vertices} the vertices $u_i$ and $v_{i}$ are the leftmost vertices of~$\Gamma_i$ and no other vertex of $G_i$ has the same $x$-coordinate as these vertices; and
    \item\label{prp:gate-right-vertices} the vertices $u_{i+1}$ and $v_{i+1}$ are the rightmost vertices of~$\Gamma_i$ and no other vertex of $G_i$ has the same $x$-coordinate as these vertices.
\end{enumerate}

We proceed as follows. \Cref{le:convex-face-given-s1,le:convex-face-given-s2} are two technical lemmas that yield special internally-strictly-convex grid drawings of cycles. We use these two lemmas to compute a block-drawing $\Gamma_i$ of each fan graph $G_i$ with appropriate area bounds. We do so by drawing each of its faces using either \cref{le:convex-face-given-s1} or \cref{le:convex-face-given-s2}, and we appropriately ``merge'' them together (\cref{le:blockdrawing}). The final drawing $\Gamma$ of~$G$ is obtained by ``gluing'' the block-drawings $\Gamma_0,\dots,\Gamma_r$ (\cref{th:outerpath}) at their shared gate edges.  

%\begin{wrapfigure}[6]{r}{0.45\textwidth}
%\vspace{-8mm}

%\end{wrapfigure}   

In the following, by \emph{half-plane of a line} we mean an open half-plane bounded by~the~line.

%The general approach to prove \cref{th:outerpath} is to compute block drawings $\Gamma_0,\dots,\Gamma_r$ of the fan graphs $G_0,\dots,G_r$ of the outerpath $G$ so that each block drawing has height $O(k)$ and width $O(n_ik)$, where $n_i$ is the number of vertices of~$G_i$ ($0\leq i \leq r)$. We iteratively construct $\Gamma$ by ``gluing'' together the block drawings $\Gamma_0,\dots,\Gamma_r$ at their common gate edges. 
%More precisely, we proceed as follows. We initialize $\Gamma$ to $\Gamma_0$. 
%Then, for $i=1,\dots,r-1$, we augment $\Gamma$ with
%a copy of~$\Gamma_{i}$ so that the drawing of the gate edge $e^*_{i}$ coincides with the drawing of~$e^*_{i}$ in the drawing constructed so far.

%We start by giving a technical lemma about a strictly-convex drawing of a single polygon. 

\begin{figure}[t!]
%  \begin{center}
\centering
\includegraphics[page=9,width=.45\textwidth]{Figures/outerpath.pdf}
%  \end{center}
 % \vspace{-5mm}
  \caption{
  Illustrations for the proof of \cref{le:convex-face-given-s1}.}\label{fig:convex-face-given-s1-slanted}
\end{figure}

\begin{lemma}\label{le:convex-face-given-s1}
    Let $C=(v_1,v_2,\dots, v_h)$ with $h\geq 3$ be a cycle. Let $s_1$ be a segment representing $e_1=(v_1,v_2)$ such that $v_1$ is at the origin $(0,0)$, $x(v_2)$ is a non-negative integer, and $y(v_2)=1$. Then~$C$ admits an internally-strictly-convex grid drawing $\Gamma_C$ such that $C$:
    \begin{enumerate}[label={\bf I.\arabic*},ref=I.\arabic*,leftmargin=22pt]
    \item \label{pro:s1}$s_1$ represents $e_1$ in $\Gamma_C$;
    \item \label{pro:vh}the vertex $v_h$ is placed at $(x(v_2)+h-2,1)$;
    \item \label{pro:side}let $\ell$ be the line that passes through $s_1$ oriented from $v_1$ towards $v_2$. Then, all the vertices $v_3,\dots,v_h$ lie in the right half-plane of~$\ell$;
    \item \label{pro:width}the width of~$\Gamma_C$ is $x(v_2)\cdot(h-2) +  \frac{(h-3)(h-2)}{2} +1$; and
    \item \label{pro:height}the height of~$\Gamma_C$ is~$h-1$.   
    \end{enumerate} 
\end{lemma}

\begin{proof}
We show how to construct the drawing $\Gamma_C$; refer to \cref{fig:convex-face-given-s1-slanted}.
Vertices $v_1$ and $v_2$ are at $(0,0)$ and $(x(v_2),1)$ in $\Gamma_C$.
For $j=3,\dots,h-1$, place $v_{j}$ at $(x(v_2)\cdot(j-1) +  \frac{(j-2)(j-1)}{2}, j-1)$.
Finally, place  vertex $v_h$ at $(x(v_2)+h-2,1)$. This completes the construction of $\Gamma_C$. 

Note that $\Gamma_C$ satisfies  \cref{pro:s1,pro:vh} by construction. 
Furthermore, observe that the $x$-coordinate of $v_{h-1}$ is $x(v_2)\cdot(h-2) +  \frac{(h-3)(h-2)}{2}$, while its $y$-coordinate is $h-2$.
The leftmost grid line intersecting $\Gamma_C$ passes through $v_0$ and thus this is the line $x=0$, whereas the rightmost grid line intersecting $\Gamma_C$ passes through $v_{h-1}$ and thus this is the line $x = x(v_2)\cdot(h-2) +  \frac{(h-3)(h-2)}{2}$. Therefore, $\Gamma_C$ has width equal to $x(v_{h-1})+1$, and \cref{pro:width} holds.
The lowest horizontal grid line intersecting $\Gamma_C$ passes through $v_0$ and thus this is the line $y=0$, whereas the highest horizontal grid line intersecting $\Gamma_C$ passes through $v_{h-1}$ and thus this is the line $y=h-2$.
Therefore, $\Gamma_C$ has height $y(v_{h-1})+1$, and \cref{pro:height} holds.

Next, we show that $\Gamma_C$ is strictly-convex. This is trivial if $h=3$. If $h>3$, the strict convexity follows by these observations:

\begin{itemize}
    \item For $i=1,\ldots,h-2$, the slope of the edge $(v_i,v_{i+1})$ is $\frac{1}{x(v_2)+i-1}$, that is, the slopes decrease with $i$, hence the interior angles at vertices $v_2,\ldots,v_{h-2}$ are smaller than $180^\circ$.
    \item By construction, the slopes of the edges $(v_1,v_h)$ and $(v_{h-2},v_{h-1})$ are both equal to $\frac{1}{x(v_2)+h-2}$. Since these edges are also of the same length, vertices $v_1,v_{h-2},v_{h-1},v_h$ form a parallelogram $\cal P$. This implies that the interior angles of $\cal P$ (and thus of $C$) at $v_{h-1}$ and at $v_{h}$ are smaller than $180^\circ$. Since the edge $(v_1,v_2)$ has positive slope, by hypothesis, and the slope of the edge $(v_1,v_h)$ is positive and smaller than the slope of $(v_1,v_2)$ by construction, the angle of $C$ at $v_1$ is also smaller than $180^\circ$.
\end{itemize}

Finally, the strict convexity of $\Gamma_C$ and the fact that the slope of the edge $(v_1,v_h)$ is smaller than the slope of the edge $(v_1,v_2)$ imply \cref{pro:side}.
\end{proof}

% \begin{proof}[Sketch]
% The drawing $\Gamma_C$ can be obtained as follows; see \cref{fig:convex-face-given-s1-slanted}.
% Vertices $v_1$ and $v_2$ are at $(0,0)$ and $(x(v_2),1)$ in $\Gamma_C$.
% For $j=3,\dots,h-1$, we place $v_{j}$ at %$(x(v_{i+1})+i,i+1)$.
% %$(x(v_{i+1})+x(v_2)+i,i+1)$.
% $(
% x(v_2)\cdot(j-1) +  \frac{(j-2)(j-1)}{2}, j-1)$.
% Finally, we place  vertex $v_h$ at $(x(v_2)+h-2,1)$. This completes the construction of $\Gamma_C$. We defer the proof that $\Gamma_C$ satisfies \cref{pro:s1,pro:vh,pro:side,pro:width,pro:height} \hlo{to the appendix.}
% \end{proof}

The following property is a consequence of the constraints of the drawings of \cref{le:convex-face-given-s1}.

\begin{property}\label{pro:shift-slanted}
Let $\Gamma_C$ be the drawing of a cycle $C=(v_1,v_2,\dots, v_h)$ produced by \cref{le:convex-face-given-s1}. Then the drawing obtained by shifting vertex $v_h$ in $\Gamma_C$ by any integer amounts to the right and/or to the bottom is an internally-strictly-convex grid drawing of $C$.
\end{property}

\begin{figure}[t!]
    \centering
    \includegraphics[page=6,width=\textwidth]{Figures/outerpath.pdf}
    \caption{Illustration for the proof of \cref{le:convex-face-given-s2}.}
    \label{fig:convex-face-given-s1-rounded}
\end{figure}

\begin{lemma} \label{le:convex-face-given-s2}
    Let $C=(v_1,v_2,\dots, v_h)$ with $h>3$ be a cycle. Let $e_1=(v_1,v_2)$ and $e_2=(v_i,v_{i+1})$ be two non-adjacent edges of $C$. Let $s_1$ be a segment that represents $(v_1,v_2)$ such that $v_1$ is at the origin $(0,0)$, $x(v_2)$ is a non-negative integer, and $y(v_2) = 1$. Then $C$ admits an internally-strictly-convex grid drawing $\Gamma_C$ such that:
    \begin{enumerate}[label={\bf J.\arabic*},ref=J.\arabic*,leftmargin=22pt]
        \item \label{prof:e1-rounded}$s_1$ represents $e_1$ in $\Gamma_C$; 
        \item \label{prof:e2-rounded}the edge $e_2$ is represented by a vertical segment such that $x(v_i)>x(v_2)$, $y(v_i)=1$, and $y(v_{i+1})=0$;
         \item \label{prof:width-rounded}the width  of~$\Gamma_C$ is $2 + \max \{ x(v_2)\cdot(i-2) + \frac{(i-3)(i-2)}{2}, \frac{(h-i-1)(h-i)}{2} \}$; and 
        \item \label{prof:height-rounded}the height  of~$\Gamma_C$ is $h-2$. 
    \end{enumerate} 
\end{lemma}

\begin{proof}
We show how to construct $\Gamma_C$; refer to \cref{fig:convex-face-given-s1-rounded}.
First, we define two cycles $C_1 = (v_1,v_2,\dots,v_{i})$ and $C_2 = (v_1,v_{h},v_{h-1},\dots,v_{i+1})$. 
Observe that the length of $C_1$ is $i$ and that the length of $C_2$ is $h-i+1$.
We apply \cref{le:convex-face-given-s1} to obtain a drawing $\Gamma_{C_1}$ of $C_1$ in which the edge $(v_1,v_2)$ is represented by $s_1$ and to obtain a drawing $\Gamma_{C_2}$ of $C_2$ in which the edge $(v_1,v_h)$ is represented by the segment connecting points $(0,0)$ and $(1,1)$. 
Let~$M$ be the maximum $x$-coordinate of a vertex in $\Gamma_{C_1}$ and $\Gamma_{C_2}$.
We leverage \cref{pro:shift-slanted} to obtain a drawing $\Gamma'_{C_1}$
 of $C_1$ by shifting the vertex~$v_i$ horizontally and rightward in $\Gamma_{C_1}$ so that its $x$-coordinate is $1+M$. Analogously, we leverage \cref{pro:shift-slanted} to obtain a drawing $\Gamma'_{C_2}$
 of~$C_2$ by shifting the vertex $v_{i+1}$ downward and rightward in $\Gamma_{C_2}$ so that its $x$-coordinate is $1+M$ and its $y$-coordinate is $0$.
 To obtain the drawing $\Gamma_C$ we then proceed as follows.
 We first initialize $\Gamma_C$ to $\Gamma'_{C_1}$. 
 %Note that, $\Gamma_C$ satisfied \cref{prof:e1-rounded}.
 Then we construct a vertically-mirrored copy $\Gamma''_{C_2}$ of $\Gamma'_{C_2}$, that is, we flip the sign of the $y$-coordinate of each vertex, and we insert $\Gamma''_{C_2}$ in $\Gamma_C$ so that the placement of $v_1$ is the same in both drawings.
 Finally, we remove from $\Gamma_C$ the drawing of the edges $(v_1,v_i)$ and $(v_1,v_{i+1})$ and draw the edge $(v_i,v_{i+1})$ as a vertical segment of unit length. 
 %Note that, $\Gamma_C$ satisfied \cref{prof:e2-rounded}.
 This concludes the construction of $\Gamma_C$.
 
 By construction, $\Gamma_C$ satisfies \cref{prof:e1-rounded} and \cref{prof:e2-rounded}.
 Since the length of $C_1$ is $i$, it follows by \cref{le:convex-face-given-s1} that $w(\Gamma_{C_1}) = x(v_2)\cdot(i-2)+\frac{(i-3)(i-2)}{2}+1$, and that $h(\Gamma_{C_1}) = i-1$.
 Also, since the length of $C_2$ is $h-i+1$, it follows by \cref{le:convex-face-given-s1} that $w(\Gamma_{C_2}) = (h-i-1)+\frac{(h-i-2)(h-i-1)}{2}+1 = \frac{(h-i-1)(h-i)}{2}+1$, and that $h(\Gamma_{C_2}) = h-i$.
The width of $\Gamma_C$ is given by $1$ plus the maximum between $w(\Gamma_{C_1})$ and $w(\Gamma_{C_2})$, which implies \cref{prof:width-rounded}.
 Also, the height of~$\Gamma_C$ is given by $h(\Gamma_{C_1})+h(\Gamma_{C_2})-1$, which implies \cref{prof:height-rounded}.

To complete the proof, it remains to prove that $\Gamma_C$ is an internally-strictly-convex drawing of $C$. By \cref{le:convex-face-given-s1}, the interior angles of $C$ at $v_2,\ldots,v_{i-1},v_{i+2},\ldots,v_h$ are strictly smaller than $180^\circ$. 
%So, it remains to show that the interior angles of $C$ at $v_1$, $v_i$, and $v_{i+1}$ are also strictly smaller than $180^\circ$.
%
Since the slope of $(v_1,v_2)$ ranges in $(0,+\infty]$ while the slope of $(v_1,v_h)$ is either $-1$ (if $h \neq i+1$) or $0$ (if $h=i+1$), the interior angle of $C$ at $v_1$ is strictly smaller than $180^\circ$. Since $v_i$ and $v_{i+1}$ are the rightmost vertices of $\Gamma_C$, it follows that each of the interior angles of $C$ at $v_i$ and $v_{i+1}$ is also strictly smaller than $180^\circ$, completing the proof.
\end{proof}

% \begin{proof}[Sketch.]
% The drawing $\Gamma_C$ can be constructed as follows; see \cref{fig:convex-face-given-s1-rounded}.
% We obtain a drawing $\Gamma_{C_1}$ of cycle $C_1 = (v_1,v_2,\dots,v_{i})$ by applying \cref{le:convex-face-given-s1} to $C_1$
% with $e_1$ and $s_1$ being the same as those in the statement. Also, we obtain a drawing $\Gamma_{C_2}$ of cycle $C_2= (v_1,v_{h},v_{h-1}, \ldots,v_{i+1})$ by applying \cref{le:convex-face-given-s1} to $C_2$ 
% with $e_1$ being $(v_1,v_h)$ and $s_1$ being the segment connecting  $(0,0)$ and $(1,1)$. 
% Let $M$ be the maximum $x$-coordinate of a vertex in $\Gamma_{C_1}$ and $\Gamma_{C_2}$.
% Then, by \cref{pro:shift-slanted}, we obtain a drawing $\Gamma'_{C_1}$
%  of $C_1$ by shifting the vertex $v_i$ horizontally and rightward in $\Gamma_{C_1}$ so that its $x$-coordinate is $1+M$. Analogously, we obtain a drawing $\Gamma'_{C_2}$
%  of $C_2$ by shifting the vertex $v_{i+1}$ downward and rightward in $\Gamma_{C_2}$ so that its $x$-coordinate is $1+M$ and its $y$-coordinate is $0$.
%  Finally, we initialize $\Gamma_C$ to $\Gamma'_{C_1}$,  insert a 
%  vertically-mirrored copy of $\Gamma'_{C_2}$ in $\Gamma_C$ so that the placement of $v_1$ is the same in both drawings, remove from $\Gamma_C$ the drawing of the edges $(v_1,v_i)$ and $(v_1,v_{i+1})$, and draw the edge $(v_i,v_{i+1})$ as a vertical segment of unit length. 
%  We defer the proof that $\Gamma_C$ satisfies \cref{prof:e1-rounded,prof:e2-rounded,prof:width-rounded,prof:height-rounded} to the~appendix.
% \end{proof}

\begin{property}\label{pro:shift-rightward}
Let $\Gamma_C$ be the drawing of a cycle $C=(v_1,v_2,\dots, v_i,v_{i+1},\dots,v_h)$ produced by \cref{le:convex-face-given-s2}. Then the drawing obtained by shifting  vertices $v_{i}$ and $v_{i+1}$ in $\Gamma_C$ by the same integer amount to the right is an internally-strictly-convex grid drawing of $C$.
\end{property}

\noindent We next show how to construct a block-drawing of a fan graph with appropriate area bounds.

\begin{lemma}\label{le:blockdrawing}
   Every $n_i$-vertex fan graph $G_i$ with internal faces of size at most $k_i$ admits a block-drawing whose width is $O(n_ik_i)$ and whose height is $O(k_i)$.
\end{lemma}

\begin{proof}
Let $G_i$ be an $n_i$-vertex fan graph with gate edges $e^*_i = (u_i,v_i)$ and $e^*_{i+1} = (u_{i+1},v_{i+1})$.
To prove the statement, we distinguish two cases based on whether $G_i$ contains internal edges or not; refer to \cref{fig:convex-blocks}. 

In the latter case (see, e.g., $G_1$ in \cref{fig:convex-blocks}), $G_i$ is a cycle $C=(u_i,v_i,\dots,v_{i+1},u_{i+1},\dots)$ with $n_i \geq 4$.
Then a block drawing $\Gamma_i$ of $G_i$ can be obtained by applying 
\cref{le:convex-face-given-s2} to cycle $C$ with $e_1$ being $(u_i,v_i)$, $e_2$ being $(v_{i+1},u_{i+1})$, and with $s_1$ being the vertical segment connecting $(0,0)$ and $(0,1)$. In particular, \cref{prof:e1-rounded,prof:e2-rounded} of \cref{le:convex-face-given-s2} immediately imply \cref{prp:gate-edges} of $\Gamma_i$. Also, \cref{prp:gate-left-vertices,prp:gate-right-vertices} follow from the strict-convexity of~$\Gamma_i$.
By~\cref{prof:width-rounded} and \cref{prof:height-rounded} of \cref{le:convex-face-given-s2}, the width of $\Gamma_i$ is in $O(k_i^2)$, which is in $O(n_ik_i)$, and the height of $\Gamma_i$ is in $O(k_i)$.

\begin{figure}[t]
    \centering
    \includegraphics[page=10,width=.8\textwidth]{Figures/outerpath.pdf}
    \caption{Illustration for  \cref{le:blockdrawing} and \cref{th:outerpath}. The drawing represents the graph in~\cref{fig:fan-decomposition}.}
    \label{fig:convex-blocks}
\end{figure}

If $G_i$ contains internal edges, we further distinguish two sub-cases based on which, between $u_i$ and $v_i$, is incident to an internal edge of $G_i$.
First, we discuss the case in which $u_i$ is incident to internal edges  (see, e.g., $G_0$ in \cref{fig:convex-blocks}). 
Let $f'$ be the internal face of $G_i$ incident to $e^*_{i+1}$ and let $f''$ be the internal face of $G_i$ that shares an edge with $f'$. We denote the edge shared by $f'$ and $f''$ as $e'$, and note that $e'$ is incident to $u_i$; also, we denote as $e''$ the edge of $f''$ distinct from $e'$ and incident to $u_i$.
We temporarily remove the edge $e'$ from $G_i$. We treat the resulting graph $G^-_i$ as a fan graph with gate edges $e^*_i$ and $e^*_{i+1}$.

If $G^-_i$ is a simple cycle, which occurs if $e'$ is the only internal edge of $G_i$ and thus~$e'' = e^*_i$, we compute a block drawing of $G_i$ as follows. 
%Note that, in this case, $n_i \leq 2k-2$.
First, we compute a block drawing~$\Gamma^-_i$ of~$G^-_i$ by means of \cref{le:convex-face-given-s2}. Then, since~$\Gamma^-_i$ is a strictly-convex polygon, we obtain a block drawing of~$G_i$ by simply drawing the edge $e'$ as a straight-line chord of this polygon. Since in this case we have $n_i \leq 2k_i-2$, by~\cref{prof:width-rounded} and \cref{prof:height-rounded} of \cref{le:convex-face-given-s2} the width of $\Gamma_i$ is in $O(k_i^2)$, which is in $O(n_ik_i)$, and the height of $\Gamma_i$ is in $O(k_i)$.

Suppose now that $G^-_i$ contains internal edges (incident to $u_i$). Let $g_1,g_2,\dots,g_{t+1}$ be the internal faces of $G^-_i$ in clockwise order around $u_i$, where~$g_1$ is incident to $e^*_i$ and~$g_{t+1}$ is incident to $e''$ and to $u_{i+1}$. Note that $g_{t+1}$ is the face resulting from ``merging''~$f'$ and~$f''$ after the removal of $e'$. For $j=2,\dots,t+1$, let $(u_i,z_j)$ be the unique edge shared by $g_{j-1}$ and $g_j$. First, we compute a block drawing~$\Gamma^-_i$ of~$G^-_i$ as follows. We initialize $\Gamma^-_i$ to a drawing obtained by applying \cref{le:convex-face-given-s1} to the cycle bounding~$g_1$ with~$e_1$ being $e^*_i$ and $s_1$ being the vertical segment of unit length between $(0,0)$ and~$(0,1)$. Next, for $j=2,\dots,t$, assume that we have an internally-strictly-convex grid drawing $\Gamma^-_i$ of the graph formed by the boundaries of the faces $g_1,g_2,\dots,g_{j-1}$. Moreover, we assume that $\Gamma^-_i$ lies in the left half-plane of the line $\ell_{j}$ passing through $u_i$ and $z_j$ when going from $u_i$ to $z_j$; also, we assume that the segment representing $(u_i,z_j)$ is such that $y(z_j)=1$ and $x(z_j)\geq 0$. We extend $\Gamma^-_i$ by adding a drawing of the cycle bounding~$g_j$ obtained by applying \cref{le:convex-face-given-s1} to such a cycle with~$e_1$ being~$(u_i,z_j)$ and~$s_1$ being the segment representing $(u_i,z_j)$ in the drawing $\Gamma^-_i$ computed so far. 
Observe that, by \cref{le:convex-face-given-s1}, the drawing of the cycle bounding $g_j$ is strictly-convex and that the edges of $g_j$ do not cross any of the edges of $\Gamma^-_i$ by \cref{pro:side}.
Therefore, $\Gamma^-_i$ is an internally-strictly-convex grid drawing. Moreover, by \cref{pro:vh} of \cref{le:convex-face-given-s1}, $\Gamma^-_i$  lies in the left half-plane of the line $\ell_{j+1}$  passing through $u_i$ and $z_{j+1}$ when going from $u_i$ to $z_{j+1}$; also, the segment representing $(u_i,z_{j+1})$ is such that $y(z_{j+1})=1$ and $x(z_{j+1})\geq 0$. Once $\Gamma^-_i$ is an internally-strictly-convex grid drawing of the graph formed by the boundaries of the faces $g_1,g_2,\dots,g_t$, we have that the edge~$e''$ belongs to $\Gamma^-_i$, as it belongs to $g_t$. Then we augment $\Gamma^-_i$ with a drawing of the cycle $C_{t+1}$ bounding the face $g_{t+1}$ obtained by applying \cref{le:convex-face-given-s2} with~$e_1$ being~$e''$,~$e_2$ being~$e^*_{i+1}$, and~$s_1$ being the segment representing $e''$ in the drawing $\Gamma^-_i$ computed so far. Furthermore, by \cref{pro:shift-rightward}, we shift (by the same amount) to the right vertices~$u_{i+1}$ and~$v_{i+1}$ so that they lie one unit to the right of the rightmost of the remaining vertices of~$G_i$.
Note that this shift might be needed when the boundaries of some of the faces $g_1,\ldots,g_t$ drawn by \cref{le:convex-face-given-s1} are ``long'' (e.g., of length $k$), while the boundary of the last face $g_{t+1}$ drawn by \cref{le:convex-face-given-s2} is ``short'' (e.g., of length $4$); in this situation, $u_{i+1}$ and~$v_{i+1}$ might not be the rightmost vertices of~$\Gamma^-_i$ and this requires shifting~$u_{i+1}$ and~$v_{i+1}$. 
Furthermore, observe that the edge $e''$ lies along~$\ell_{t+1}$ and that $C_{t+1}$ is drawn as a  strictly-convex polygon $P_{t+1}$ with one side along~$\ell_{t+1}$ (\cref{prof:e1-rounded}) that lies to the right of this line (\cref{prof:e2-rounded}). Therefore, $\Gamma^-_i$ is planar, and thus it is an internally-strictly-convex grid drawing of $G^-_i$.
Finally, to obtain the desired block drawing $\Gamma_i$ of $G_i$, we simply draw the edge $e'$ as a straight-line chord of $P_{t+1}$ in $\Gamma^-_i$.

Since all the applications of \cref{le:convex-face-given-s1} and \cref{le:convex-face-given-s2} use segments incident to $u_i$, it follows that the height of the computed drawing is in $O(k_i)$. 

%Further, since each application of \cref{le:convex-face-given-s1} and \cref{le:convex-face-given-s2} to a face $f$ yields an increment of the width of the computed drawing by $O(|f|k)$ it follows that the total width of the drawing is $O(nk)$\todo{This has been forgotten, I guess.}.
%
% To estimate an upper bound on the width of $\Gamma_i$, it suffices to compute the width of the drawing of  face $f_{t+1}$ drawn using \cref{le:convex-face-given-s2} while assuming that such a face has been obtained as the merge of two faces of length at most $k_i$. To do so, we first upper bound the $x$-coordinate of vertex $z_{t+1}$. 
% %
% By \cref{pro:vh} of \cref{le:convex-face-given-s1}, we have that the $x$-coordinate of vertex $z_{i}$ is at most $\sum^{i-1}_{j=1} |f_j|$, with $2 \leq j \leq t+1$. Therefore, $x(z_{t+1}) \leq n_i$. It follows by \cref{prof:width-rounded} of \cref{le:convex-face-given-s2}, that the width of the drawing of $f_{t+1}$ is at most
% $1 + \max \{ x(z_{t+1})\cdot(i-2) + \frac{(i-3)(i-2)}{2}, \frac{(h-i-1)(h-i)}{2} \}$, with $i\leq k_i$ and $h \leq 2k_i$. Since $x(z_{t+1}) \leq n_i$, we get that the width of $\Gamma_i$ is at most $x(z_{t+1})\cdot(i-2) + \frac{(i-3)(i-2) + (h-i-1)(h-i)}{2} \leq   n_i \cdot k_i + k_i^2 \in O(n_i k_i)$ because $n_i \geq k_i$. 

We now prove an upper bound on the width of $\Gamma_i$. We first upper bound the $x$-coordinate of any vertex $z_{p}$, with $2 \leq p \leq t+1$. By \cref{pro:vh} of \cref{le:convex-face-given-s1}, we have that the $x$-coordinate of $z_{p}$ is less than $\sum^{p-1}_{j=1} |g_p|$, which is at most $n_i$. It follows by \cref{pro:width} of \cref{le:convex-face-given-s1} that the width of the drawing of the face $g_j$, for any $j=2,\dots,t$, is less than $n_ik_i +  k_i^2/2 +1$. Also, it follows by \cref{prof:width-rounded} of \cref{le:convex-face-given-s2}, that the width of the drawing of the face $g_{t+1}$, before the shift of $u_{i+1}$ and~$v_{i+1}$, is less than $2n_ik_i +  2k_i^2 +2$; note that the face $g_{t+1}$ is the merge of two faces of length at most $k_i$. Thus, we get that the width of $\Gamma_i$ is in $O(n_i k_i+ k_i^2)$, and hence in $O(n_i k_i)$.

% We now prove an upper bound on the width of $\Gamma_i$. Since $u_{i+1}$ and~$v_{i+1}$ are the rightmost vertices of $\Gamma_i$, it suffices to compute the width of the drawing of the face $g_{t+1}$ drawn using \cref{le:convex-face-given-s2}, which is the merge of two faces of length at most $k_i$. To do so, we first upper bound the $x$-coordinate of vertex $z_{t+1}$. 
% %
% By \cref{pro:vh} of \cref{le:convex-face-given-s1}, we have that the $x$-coordinate of a vertex $z_{p}$ is less than $\sum^{p-1}_{j=1} |g_p|$, with $2 \leq p \leq t+1$. Therefore, $x(z_{t+1}) \leq n_i$. It follows by \cref{prof:width-rounded} of \cref{le:convex-face-given-s2}, that the width of the drawing of $g_{t+1}$ is at most $2 + \max \{ x(z_{t+1})\cdot(a-2) + \frac{(a-3)(a-2)}{2}, \frac{(h-a-1)(h-a)}{2} \}$, with $a\leq k_i$ and $h \leq 2k_i$. Since $x(z_{t+1}) \leq n_i$, we get that the width of $\Gamma_i$ is in $O(n_i k_i+ k_i^2)$, and hence in $O(n_i k_i)$.

Finally, the case in which $v_i$ is incident to an internal edge (see, e.g., $G_2$ in \cref{fig:convex-blocks}) can be handled by first flipping the outerplanar embedding of $G_i$, afterward exchanging the roles of $u_i$ and $v_i$, then computing a block drawing $\Gamma_i$ of $G_i$ as described above, and finally vertically flipping $\Gamma_i$ around the axis $y:=0.5$.
\end{proof}

We are now ready to prove the main theorem of this section. 

\begin{theorem} \label{th:outerpath}
Every $n$-vertex outerpath whose internal faces have size at most $k$ admits an embedding-preserving internally-strictly-convex grid drawing in $O(nk^2)$ area.
%Every $n$-vertex outerplane graph whose faces have size at most $k$ and whose weak dual graph is a path admits a strictly-convex grid drawing in $O(nk^2)$ area.
\end{theorem}
\begin{proof}
   We compute block drawings $\Gamma_0,\dots,\Gamma_r$ of the fan graphs $G_0,\dots,G_r$ of the outerpath~$G$ so that each block drawing has height $O(k_i)$ and width $O(n_ik_i)$ by \cref{le:blockdrawing}, where $n_i$ is the number of vertices of~$G_i$ and $k_i$ is the maximum size of any face of $G_i$ ($0\leq i \leq r)$. We iteratively construct $\Gamma$ by ``gluing'' together the block drawings $\Gamma_0,\dots,\Gamma_r$ at their common gate edges. 
More precisely, we proceed as follows. We initialize $\Gamma$ to $\Gamma_0$. 
Then, for $i=1,\dots,r$, we augment $\Gamma$ with
a copy of~$\Gamma_{i}$ so that the drawing of the gate edge $e^*_{i}$ in~$\Gamma_{i}$ coincides with the drawing of~$e^*_{i}$ in the drawing constructed so far. It follows that the height of $\Gamma$ is $O(\underset{0\leq i \leq r}{\max}\{k_i\})$, which is in $O(k)$, while its width is $O(\sum_{i=0}^{r} n_ik_i)$, which is in $O(nk)$.
\end{proof}

As a corollary of \cref{th:outerpath}, we obtain the following.

\begin{corollary}
Let $G$ be an $n$-vertex outer-$1$-plane graph whose outer frame is an outerpath. Then $G$ admits an embedding-preserving straight-line grid drawing in $O(n)$ area.
\end{corollary}

\begin{proof}
Let $G'$ be the outer frame of $G$ and let $(f_1,f_2,\dots,f_p)$ be the weak dual tree of $G'$, which is a path. Recall that, by the construction of the outer frame, every face of $G'$ that contains a pair of crossing edges in $G$ is delimited by a $4$-cycle. We triangulate all the other internal faces of $G'$, while maintaining $G'$ an outerpath. This is done as follows. For each face $f_i$ that does not contain any pair of crossing edges in $G$, let $(u_1,u_2,\dots,u_k)$ be the cycle bounding $f_i$. Also, let $(u_1,u_2)$ and $(u_j,u_{j+1})$ be the edges of $G'$ dual to the edges $(f_{i-1},f_i)$ and $(f_{i},f_{i+1})$, respectively; if $i=1$, the edge $(u_1,u_2)$ is chosen arbitrarily among the edges incident to $f_1$ and different from $(u_j,u_{j+1})$, and similarly for the case in which $i=p$. We add inside $f_i$ the edges $(u_1,u_3),(u_1,u_4),\dots,(u_1,u_j)$ and the edges $(u_j,u_k),(u_j,u_{k-1}),\dots,(u_j,u_{j+2})$. After the augmentation, every internal face of $G'$ has at most $4$ vertices, hence by \cref{th:outerpath} we have that $G'$ admits an internally-strictly-convex drawing $\Gamma'$ in $O(n)$ area. Since the internal faces of $G'$ are represented as strictly-convex polygons in $\Gamma'$, re-introducing the removed crossing edges in $\Gamma'$ as straight-line segments, and removing the edges of $G'$ that do not belong to $G$ results in an embedding-preserving straight-line grid drawing of $G$ in $O(n)$ area. 
\end{proof}

%By \cref{th:outerpath}, \hlg{the outerpath $G'$ obtained from $G$ by removing its crossing edges admits an internally-strictly-convex drawing $\Gamma'$ in $O(nk^2)$ area. Since the internal faces of $G'$ are represented as strictly-convex polygons in $\Gamma'$, re-introducing the removed crossing edges in $\Gamma'$ as straight-line segments results in an embedding-preserving straight-line grid drawing $\Gamma$ of $G$ in $O(nk^2)$ area. 

\subsection{Outerplanar Graphs of Bounded Diameter}\label{sse:diameter}

In this section we prove bounds on the area requirements of internally-strictly-convex grid drawings of outerplane graphs of bounded diameter. Let $G$ be a connected graph. The \emph{distance} $\delta_G(u,v)$ between two vertices $u$ and $v$ in $G$ is the length (number of edges) of the shortest path between $u$ and $v$ in $G$. The \emph{diameter} of $G$ is the maximum distance between any two vertices of $G$. We start with the following lower bound.

%For a vertex $u$ in $G$, the \emph{$u$-diameter} is the maximum distance between $u$ and any other vertex of $G$. Observe that, for any vertex $u$ in $G$, the $u$-diameter is at most the diameter of $G$, which in turn is at most twice the $u$-diameter. 

\begin{theorem} \label{th:lower-diameter}
For every integers $n\geq 3$ and $1\leq d\leq n/2$, there exists an $n$-vertex outerplane graph with diameter $d$ that requires $\Omega(nd^2)$ area in any internally-strictly-convex grid drawing.
\end{theorem}

\begin{proof}
Assume that $d\geq 3$, as if $d\leq 2$ every $n$-vertex outerplane graph with diameter $d$ requires $\Omega(nd^2)\in \Omega(n)$ area in any internally-strictly-convex grid drawing.

\begin{figure}[htb]
    \begin{subfigure}[b]{0.45\textwidth}
      \centering
      \includegraphics[page=1]{Figures/diameter.pdf}
      \subcaption{}
      \label{fig:diameter-even}
    \end{subfigure}
    \hfill
    \begin{subfigure}[b]{0.45\textwidth}
      \centering
      \includegraphics[page=2]{Figures/diameter.pdf}
      \subcaption{}
      \label{fig:diameter-odd}
    \end{subfigure}
    \caption{Illustrations for the proof of \cref{th:lower-diameter}. (a) $n=12$ and $d=4$. (b) $n=14$ and $d=5$.}
    \label{fig:diameter}
\end{figure}

If $d$ is even, as in \cref{fig:diameter-even}, then let~$G$ consist of~$\lfloor \frac{n-1}{d-1} \rfloor$ cycles with~$d$ vertices, all sharing the same vertex~$v$, plus additional~$n-1-(d-1)\cdot \lfloor \frac{n-1}{d-1} \rfloor$ vertices, all adjacent to~$v$. Note that~$\lfloor \frac{n-1}{d-1} \rfloor\geq 2$, since $d\geq 3$ and $n\geq 2d$. If~$d$ is odd, as in \cref{fig:diameter-odd}, then let~$G$ consist of a cycle with~$d+1$ vertices, together with~$\lfloor \frac{n-(d+1)}{d-1} \rfloor$ cycles with~$d$ vertices, where again all the cycles share the same vertex~$v$, plus additional~$n-(d+1)-(d-1)\cdot \lfloor \frac{n-(d+1)}{d-1} \rfloor$ vertices, all adjacent to~$v$. Note that~$\lfloor \frac{n-(d+1)}{d-1} \rfloor\geq 1$, since $d\geq 3$ and $n\geq 2d$. Clearly,~$G$ has~$n$ vertices and diameter~$d$. As in the proof of \cref{th:lower}, we have that any internally-strictly-convex grid drawing of~$G$ must be outerplanar and hence such a drawing contains~$\Omega(\frac{n}{d})$ internal faces, each of which has size~$\Omega(d)$ and hence occupies~$\Omega(d^3)$ area \cite{Andrews1963}, from which the lower bound of the theorem follows.
\end{proof}

We now show an upper bound matching the lower bound of \cref{th:lower-diameter}.

\begin{theorem} \label{th:upper-diameter}
Every $n$-vertex outerplane graph with diameter $d$ admits an embedding-preserving internally-strictly-convex grid drawing in $O(nd^2)$ area.
\end{theorem}

In the rest of the section we prove \cref{th:upper-diameter}. 

Let $G$ be an outerplane graph with diameter $d$. Since adding edges to $G$ does not increase its diameter, and since the outer face is not required to be convex in an internally-strictly-convex grid drawing, the following observation allows us to assume that $G$ is biconnected. 

\begin{observation}\label{obs:assumption-biconnected}
It is possible to add edges in the outer face of $G$ so that it becomes biconnected and remains outerplane. 
\end{observation}

Next, we prove the following lemma.

\begin{lemma}\label{le:assumption-faces}
It is possible to remove some internal edges from~$G$ so that the resulting graph~$G'$ satisfies the following properties:
\begin{itemize}
    \item the diameter of $G'$ is at most $4d$; and
    \item every internal face of $G'$ has at least $4$ incident vertices.
\end{itemize}
\end{lemma}
\begin{proof}
We remove edges from $G$ in two steps. 

In a first step, while there exist two internal faces that are incident to~$3$ vertices each and that share an edge on their boundary, we remove such an edge. Overall, these removals at most double the diameter of~$G$. Indeed, denote by~$G^*$ the graph obtained after the removals and consider any path~$P=(u_1,u_2,\dots,u_k)$ in~$G$. We construct a connected subgraph~$H^*$ of~$G^*$ that contains~$u_1$ and~$u_k$ and that has at most $2k-2$ edges, as follows. For~$i=1,\dots,k-1$, if the edge~$(u_i,u_{i+1})$ is in~$G^*$, we add it to~$H^*$. Otherwise,~$(u_i,u_{i+1})$ is shared by two internal faces~$f$ and~$g$ of~$G$, with boundaries~$(u_i,u_{i+1},u(f))$ and~$(u_i,u_{i+1},u(g))$, respectively; we then add~the path~$(u_i,u(f),u_{i+1})$ to~$H^*$. Note that the path~$(u_i,u(f),u_{i+1})$ is in~$G^*$, since after the removal of~$(u_i,u_{i+1})$ the face resulting from the union of~$f$ and~$g$ has~$4$ incident vertices, hence its incident edges are not removed. The number of edges of~$H^*$ is at most twice the length of~$P$, hence at most $2k-2$. Thus,~$H^*$ contains a path from~$u_1$ to~$u_k$ whose length is at most twice the length of~$P$.

%; here by \emph{walk} in~$G^*$ we mean a sequence of edges that might pass through the same vertex more than once, such that edges that are consecutive in the sequence share an end-vertex

For the second step, note that $G^*$ does not contain any two faces that share an edge on their boundary and that are incident to~$3$ vertices each. Then, while $G^*$ contains a face incident to~$3$ vertices, we remove an arbitrary edge incident to it. Overall, these removals at most double the diameter of~$G^*$. Indeed, denote by~$G'$ the graph obtained after the removals and consider any path~$P=(u_1,u_2,\dots,u_k)$ in~$G'$. We construct a connected subgraph~$H'$ of~$G'$ that contains~$u_1$ and~$u_k$ and that has at most $2k-2$ edges, as follows. For~$i=1,\dots,k-1$, we add the edge~$(u_i,u_{i+1})$ to~$H'$ if it belongs to~$G'$. Otherwise,~$(u_i,u_{i+1})$ is incident to a face~$f$ of $G^*$ delimited by a cycle~$(u_i,u_{i+1},u(f))$, and we add to~$H'$ the path~$(u_i,u(f),u_{i+1})$. Such a path is in~$G'$, since the faces adjacent to $f$ have more than $3$ incident vertices, hence the edges incident to $f$ can only be removed because of $f$. The number of edges of~$H'$ is at most twice the length of~$P$, hence at most $2k-2$. Thus,~$H'$ contains a path from~$u_1$ to~$u_k$ whose length is at most twice the length of~$P$.
\end{proof}

Note that, inserting the edges removed from $G$ in an internally-strictly-convex grid drawing of $G'$, results in an internally-strictly-convex grid drawing of $G$. Hence, by \cref{obs:assumption-biconnected} and \cref{le:assumption-faces}, in order to prove \cref{th:upper-diameter} we can hereafter assume that $G$ is an $n$-vertex biconnected outerplane graph with diameter $d$ and that every internal face of $G$ has at least $4$ vertices.

In the following, we say that a straight-line segment~$\overline{pq}$ is an \emph{$s \times t$-segment} if~$|x(p)-x(q)|=s$ and~$|y(p)-y(q)|=t$. We now define two types of internally-strictly-convex drawings. We will then show an algorithm that recursively constructs such drawings within small area.

Let~$m$ be the number of edges of~$G$, let~$a$ and~$b$ be two positive integers such that~$b-a+1=m$, that is, the interval $[a,b]$ contains $m$ integers, and let~$(u,v)$ be an edge incident to the outer face of~$G$, where~$u$ comes right before~$v$ in counter-clockwise order along the outer face of~$G$. 
A \emph{$(u,v,a,b)$-top-drawing}~$\Gamma$ of~$G$ is an internally-strictly-convex grid drawing of~$G$ satisfying the following properties (see \cref{fig:top-drawing}):
\begin{enumerate}[label=\textbf{T\arabic*}] 
    \item First, the edge~$(u,v)$ is represented by an~$s\times 1$-segment~$\sigma$, with~$a\leq s\leq b$, where~$u$ is to the left of~$v$.
    \item
    Second, consider the half-line starting at~$u$ with slope~$1/a$, and the half-line starting at~$v$ with slope~$1/b$. Such half-lines, together with the segment~$\sigma$, partition the plane in two closed regions. 
    Then~$\Gamma$ is contained in the closed region above~$\sigma$.
\end{enumerate}

\begin{figure}[htb]
    \begin{subfigure}[b]{0.45\textwidth}
      \centering
      \includegraphics[page=3]{Figures/diameter.pdf}
      \subcaption{}
      \label{fig:top-drawing}
    \end{subfigure}
    \hfill
    \begin{subfigure}[b]{0.5\textwidth}
      \centering
      \includegraphics[page=4]{Figures/diameter.pdf}
      \subcaption{}
      \label{fig:side-drawing}
    \end{subfigure}
    \caption{Schematic visualizations of (a) a $(u,v,a,b)$-top-drawing and (b) a $(u,v,\Sigma,a,b)$-side-drawing.}
    \label{fig:drawing-types}
\end{figure}

Now, let $m$, $a$, $b$, and~$(u,v)$ be defined as above, and additionally let $\Sigma$ be a set that consists either of one $s\times 1$-segment $\sigma$ with $0<s<a$, or of two segments $\sigma_1$ and $\sigma_2$, where $\sigma_1$ is a $z_1\times 1$-segment, $\sigma_2$ is a $z_2\times 1$-segment, and $0<z_1<z_2<a$; in the latter case, let $\sigma$ be a $(z_1+z_2)\times 2$-segment.  
A \emph{$(u,v,\Sigma,a,b)$-side-drawing}~$\Gamma$ of~$G$ is an internally-strictly-convex grid drawing of~$G$ satisfying the following properties (see \cref{fig:side-drawing}):
\begin{enumerate}[label=\textbf{S\arabic*}] 
    \item First, the edge~$(u,v)$ is represented by~$\sigma$, with~$v$ to the left of~$u$.
    \item Second, consider the half-line starting at~$u$ with slope~$1/a$, and the half-line starting at~$v$ with slope~$1/b$. Such half-lines, together with the segment~$\sigma$, partition the plane in two closed regions. Then~$\Gamma$ is contained in the closed region to the right of~$\sigma$.
\end{enumerate}

Note that, differently from a $(u,v,a,b)$-top-drawing, in a $(u,v,\Sigma,a,b)$-side-drawing the representation of $(u,v)$ is prescribed.

We now prove, by simultaneous induction on $m$, that $G$ admits:
\begin{itemize}
\item a $(u,v,a,b)$-top-drawing such that, for any vertex $w\neq u$, the difference $y(w)-y(u)$ is at most $2\cdot \delta_G(u,w)$ and the difference $x(w)-x(u)$ is at most $2b\cdot \delta_G(u,w)$; and
\item a $(u,v,\Sigma,a,b)$-side-drawing such that, for any vertex $w\neq v$, the difference $y(w)-y(v)$ is at most $2\cdot \delta_G(v,w)$ and the difference $x(w)-x(v)$ is at most $2b\cdot \delta_G(v,w)$.
\end{itemize}
We call the constraints on the width and on the height of such drawings the \emph{width} and \emph{height properties}, respectively. By choosing $a=1$ and $b=m$, since $m \in O(n)$ and $\delta_G(u,v)\leq d$, we have that the height of each drawing is in $O(d)$ and the width is in $O(nd)$, hence \cref{th:upper-diameter} follows.

In the base case, we have $m=1$. 
\begin{itemize}
    \item A $(u,v,a,b)$-top-drawing $\Gamma$ is constructed by placing $u$ at any grid point of the plane and by placing $v$ so that $x(v)-x(u)=a$ and $y(v)-y(u)=1$. Obviously, $\Gamma$ is an internally-strictly-convex grid drawing of $G$. Since $m=1$, we have $b=a$, and hence Property T1 of a $(u,v,a,b)$-top-drawing is satisfied. Property T2 is also trivially satisfied, given that $G$ consists of $(u,v)$ only. By construction, $y(v)-y(u)=1<2=2\cdot \delta_G(u,v)$ and $x(v)-x(u)=a<2b=2b\cdot \delta_G(u,v)$.
    \item A $(u,v,\Sigma,a,b)$-side-drawing $\Gamma$ is constructed as $\Gamma=\sigma$, where the leftmost end-point of $\sigma$ represents $v$ and the rightmost end-point of $\sigma$ represents $u$. Obviously, $\Gamma$ is an internally-strictly-convex grid drawing of $G$ satisfying Property S1. Property S2 is also trivially satisfied, given that $G$ consists of $(u,v)$ only. By the definition of $\sigma$, we have $y(u)-y(v)\leq 2=2\cdot \delta_G(v,u)$ and $x(u)-x(v)<2a\leq 2b=2b\cdot \delta_G(v,u)$. 
\end{itemize}  

Assume now that a $(u,v,a,b)$-top-drawing and a $(u,v,\Sigma,a,b)$-side-drawing exist if $G$ has at most $m-1$ edges. We show how to construct such drawings if $G$ has $m$ edges. Let $f$ be the unique internal face of $G$ the edge $(u,v)$ is incident to, and let $u_1=u,u_2,\dots,u_k=v$ be the clockwise order of the vertices along the boundary of $f$. For $i=1,\dots,k-1$, let $G_i:=G[u_{i+1},u_i]$ be the outerplane subgraph of $G$ induced by $u_i$, by $u_{i+1}$, and by the vertices that are encountered when traversing the boundary of the outer face of $G$ in counter-clockwise direction from $u_{i+1}$ to $u_i$. For $i=1,\dots,k-1$, let $m_i\geq 1$ be the number of edges of $G_i$; note that $\sum_{i=1}^{k-1} m_i=m-1$, given that every edge of $G$ belongs to a single graph $G_i$, except for the edge $(u,v)$, which does not belong to any graph $G_i$. Finally, for $i=1,\dots,k-1$, let $a_i=a+\sum_{j=1}^{i-1} m_j$ and $b_i=a-1+\sum_{j=1}^{i} m_j$. Note that $b_i-a_i+1=m_i$ and that $a=a_1<b_1<a_2<b_2<\dots<a_{k-1}<b_{k-1}<b$. 

\begin{figure}[htb]
    \begin{subfigure}[b]{0.9\textwidth}
      \centering
      \includegraphics[page=5,scale=.9]{Figures/diameter.pdf}
      \subcaption{In this example, we have $k=6$.}
      \label{fig:top-drawing-construction-even}
    \end{subfigure}
    \\
    \begin{subfigure}[b]{0.9\textwidth}
      \centering
      \includegraphics[page=6,scale=.9]{Figures/diameter.pdf}
      \subcaption{In this example, we have $k=5$.}
      \label{fig:top-drawing-construction-odd}
    \end{subfigure}
    \caption{Construction of a $(u,v,a,b)$-top-drawing if $k$ is even (a) or odd (b). In this and the following figures, in order to improve the readability, the vertices do not lie at grid points and the edges are steeper than they should be. The little symbols next to the edges match segments that are equivalent under translation.}
    \label{fig:top-drawing-construction}
\end{figure}

The basic idea of our construction is the following. We would like to draw $G$ so that its edges have slopes $\frac{1}{a},\frac{1}{a+1},\frac{1}{a+2},\dots,\frac{1}{b}$. Hence, we will partition the interval of integers $[a,b]$ into sub-intervals $[a_1,b_1], [a_2,b_2], \dots, [a_{k-1},b_{k-1}]$ and assign such intervals to the subgraphs $G_1,G_2,\dots,G_{k-1}$, respectively, so to draw the edges of each subgraph $G_i$ with slopes $\frac{1}{a_i},\frac{1}{a_i+1},\frac{1}{a_i+2},\dots,\frac{1}{b_i}$. One complication to this strategy is that the path $(u_1,u_2,\dots,u_{k})$ on the boundary of $f$ cannot be drawn with positive, and progressively decreasing, slopes, as otherwise $x(u_k)-x(u_1)$ would be in $\Omega(nd)$, while $\delta_G(u_1,u_k)$ is only in $O(1)$, and hence this would not result in a drawing with the area bounds we aim for. Thus, we need to split $(u_1,u_2,\dots,u_{k})$ into two paths of roughly equal length, each to be drawn with positive slopes, that are progressively decreasing in the first path and progressively increasing in the second one. The subgraphs attached to the first path have a top-drawing, while the ones attached to the second path have a side-drawing. The integers in the interval $[a,b]$ are still distributed among the subgraphs $G_1,G_2,\dots,G_{k-1}$, however the edges of the second path use the same slopes as the edges of the first path, with an exception that might arise based on the parity of $k$. In fact, we might require that the last edge of the second path has the $x$- and $y$-extensions of the first two edges of the first path combined, or that the first edge of the first path has the $x$- and $y$-extensions of the last two edges of the second path combined. This results in  edges with $y$-extension equal to $2$, rather than $1$ as all the other edges of the paths. We now formalize our construction.

We first show how to construct a $(u,v,a,b)$-top-drawing $\Gamma$ of $G$, refer to \cref{fig:top-drawing-construction}. For $i=1,\dots,\lceil \frac{k}{2}\rceil$, recursively construct a $(u_i,u_{i+1},a_i,b_i)$-top-drawing $\Gamma_i$ of $G_i$. Recall that, by definition of top-drawing, for $i=1,\dots,\lceil \frac{k}{2}\rceil$, we have that the edge $(u_i,u_{i+1})$ is represented by an~$s_i\times 1$-segment~$\sigma_i$, with~$a_i\leq s_i\leq b_i$, where~$u_i$ is to the left of~$u_{i+1}$. The drawings $\Gamma_2,\dots,\Gamma_{\lceil \frac{k}{2}\rceil}$ are translated so that the placement of $u_i$ coincides in $\Gamma_i$ and $\Gamma_{i-1}$. Now the construction differs if $k$ is even or odd. 
\begin{itemize}
    \item If $k$ is even, as in \cref{fig:top-drawing-construction-even}, for $i=\frac{k}{2}+1,\dots,k-1$, recursively construct a $(u_i,u_{i+1},\Sigma_i,a_i,b_i)$-side-drawing $\Gamma_i$ of $G_i$, where $\Sigma_i$ contains only the segment $\sigma_{i-\frac{k}{2}}$, and translate it so that the placement of $u_i$ coincides in $\Gamma_i$ and $\Gamma_{i-1}$. Note that $\sigma_{i-\frac{k}{2}}$ is an $s_{i-\frac{k}{2}}\times 1$-segment with $s_{i-\frac{k}{2}}<b_{i-\frac{k}{2}}<a_i$, hence  the hypotheses for the construction of a side-drawing are satisfied.
    \item If $k$ is odd, as in \cref{fig:top-drawing-construction-odd}, recursively construct a $(u_{\frac{k+3}{2}},u_{\frac{k+5}{2}},\Sigma_{\frac{k+3}{2}},a_{\frac{k+3}{2}},b_{\frac{k+3}{2}})$-side-drawing of $G_{\frac{k+3}{2}}$, where $\Sigma_{\frac{k+3}{2}}$ contains an $s_1\times 1$-segment and an $s_2\times 1$-segment, and translate it so that the placement of $u_{\frac{k+3}{2}}$ coincides in $\Gamma_{\frac{k+3}{2}}$ and $\Gamma_{\frac{k+1}{2}}$. Note that $0<s_1<s_2\leq b_2<a_{\frac{k+3}{2}}$, hence  the hypotheses for the construction of a side-drawing are satisfied. Also, for $i=\frac{k+5}{2},\dots,k-1$, recursively construct a $(u_i,u_{i+1},\Sigma_i,a_i,b_i)$-side-drawing of $G_i$, where $\Sigma_i$ only contains the segment $\sigma_{i-\frac{k-1}{2}}$, and translate it so that the placement of $u_i$ coincides in $\Gamma_i$ and $\Gamma_{i-1}$. Note that $\sigma_{i-\frac{k-1}{2}}$ is an $s_{i-\frac{k-1}{2}}\times 1$-segment with $0<s_{i-\frac{k-1}{2}}<b_{i-\frac{k-1}{2}}<a_i$, hence the hypotheses for the construction of a side-drawing are satisfied.
\end{itemize}

This completes the construction of $\Gamma$. We now prove that $\Gamma$ satisfies the required properties.

First, note that the segment $\sigma$ representing the edge $(u,v)$ in $\Gamma$ is an $s_{\lceil \frac{k}{2}\rceil}\times 1$-segment. This is because, if $k$ is even, then the segments $\sigma_{\frac{k}{2}+1},\dots,\sigma_{k-1}$ are equivalent under translation to $\sigma_{1},\dots,\sigma_{\frac{k}{2}-1}$, respectively; while if $k$ is odd, then the segments $\sigma_{\frac{k+5}{2}},\dots,\sigma_{k-1}$ are equivalent under translation to $\sigma_{3},\dots,\sigma_{\frac{k-1}{2}}$, respectively,  and $\sigma_{\frac{k+3}{2}}$ occupies the same $x$- and $y$-intervals of the segments $\sigma_1$ and $\sigma_2$ combined. Note that this argument uses the fact that $k\geq 4$, which implies that the edge $(u_{\lceil \frac{k}{2}\rceil+1}, u_{\lceil \frac{k}{2}\rceil+2})$ is not the same edge as $(u,v)$. We have $a=a_1<a_{\lceil \frac{k}{2}\rceil}\leq s_{\lceil \frac{k}{2}\rceil}\leq b_{\lceil \frac{k}{2}\rceil}<b_{k-1}<b$, hence Property T1 of a top drawing is satisfied.
By construction and since $\Gamma_1,\dots,\Gamma_{k-1}$ are straight-line grid drawings, it follows that $\Gamma$ is a straight-line grid drawing. 
Since the drawings $\Gamma_1,\dots,\Gamma_{k-1}$ are planar, in order to prove the planarity of $\Gamma$, it suffices to prove that such drawings do not intersect each other and do not intersect $\sigma$. By Property T2 of a top drawing, for any $1\leq i< j\leq \lceil \frac{k}{2}\rceil$, we have that $\Gamma_i$ and $\Gamma_j$ lie on different sides of any line with a slope between $\frac{1}{b_i}$ and $\frac{1}{a_{j}}$ passing through $u_{i+1}$. Analogously, by Property S2 of a side drawing, for any $\lceil \frac{k}{2}\rceil+1\leq i< j\leq k-1$, we have that $\Gamma_i$ and $\Gamma_j$ lie on different sides of any line with a slope between $\frac{1}{b_i}$ and $\frac{1}{a_{j}}$ passing through $u_{i+1}$; by the same property, for any $\lceil \frac{k}{2}\rceil+1\leq i\leq k-1$, we have that $\Gamma_i$ and $\sigma$ lie on different sides of the vertical line through $u_{i+1}$. Consider the half-line starting from $u_{\lceil \frac{k}{2}\rceil+1}$ and passing through $u$, and consider any half-line starting from $u_{\lceil \frac{k}{2}\rceil+1}$ with a slope  between $\frac{1}{b_{\lceil \frac{k}{2}\rceil}}$ and $\frac{1}{a_{\lceil \frac{k}{2}\rceil+1}}$. The drawings $\Gamma_1,\dots,\Gamma_{\lceil \frac{k}{2}\rceil}$ lie in the wedge above the union of such half-lines, while the drawings $\Gamma_{\lceil \frac{k}{2}\rceil+1},\dots,\Gamma_{k-1}$  and the segment $\sigma$ lie in the wedge below the union of such half-lines, hence the drawings $\Gamma_1,\dots,\Gamma_{\lceil \frac{k}{2}\rceil}$  do not cross the drawings $\Gamma_{\lceil \frac{k}{2}\rceil+1},\dots,\Gamma_{k-1}$ and the segment $\sigma$. 
The fact that $\Gamma$ is internally-strictly-convex follows from that fact that so are $\Gamma_1,\dots,\Gamma_{k-1}$ and from the fact that $s_1<s_2<\dots<s_{\lceil \frac{k}{2}\rceil}$; the latter implies that the path $(u=u_1,u_2,\dots,u_{\lceil \frac{k}{2}\rceil+1})$ is represented by a sequence of segments with positive and strictly-decreasing slopes, and that the path $(v=u_{k},u_{k-1},\dots,u_{\lceil \frac{k}{2}\rceil+1})$ is represented by a sequence of segments with positive and strictly-increasing slopes. Property T2 follows from the fact that, for $i=1,\dots,\lceil \frac{k}{2}\rceil$, the drawing $\Gamma_i$ is a $(u_i,u_{i+1},a_i,b_i)$-top-drawing and hence satisfies Property T2, where $a\leq a_i\leq b_i<b$, and from the fact that,  for $i=\lceil \frac{k}{2}\rceil+1,\dots,k-1$, the drawing $\Gamma_i$ is a $(u_i,u_{i+1},\Sigma_i,a_i,b_i)$-side-drawing and hence satisfies Property S2, where $a\leq s_{i-\lfloor \frac{k}{2}\rfloor}\leq b_{\lceil \frac{k}{2}\rceil}<a_i\leq b_i<b$ if $\Sigma_i$ only contains an $s_i\times 1$-segment, and where $a\leq s_1<s_2\leq b_{\lceil \frac{k}{2}\rceil}<a_i\leq b_i<b$ if $\Sigma_i$ contains an $s_1\times 1$-segment and an $s_2\times 1$-segment. 

It remains to discuss the width and height properties of $\Gamma$. Consider any vertex $w$ of $G$ different from $u$. 
\begin{itemize}
\item First, suppose that $w=u_i$, for some $2\leq i\leq \lceil \frac{k}{2}\rceil$. Note that $\delta_G(u,w)=i-1$, since $(u=u_1,u_2,\dots,u_i=w)$ is the shortest path between $u$ and $w$ in $G$. By construction, we have $y(w)-y(u)=i-1<2i-2=2\cdot \delta_G(u,w)$. Also, we have $x(w)-x(u)=\sum_{j=1}^{i-1} s_j\leq \sum_{j=1}^{i-1} b_j<(i-1) \cdot b < 2b \cdot (i-1)=2b\cdot \delta_G(u,w)$.
\item Second, suppose that $w=u_i$, with $i=\lceil \frac{k}{2}\rceil+1$. If $k$ is even, then again $\delta_G(u,w)=i-1$, since $(u=u_1,u_2,\dots,u_i=w)$ is a shortest path between $u$ and $w$ in $G$, and the proof that the width and height properties are satisfied is the same as in the first case. However, if $k$ is odd, then $\delta_G(u,w)=i-2$, since $(u,u_k,u_{k-1},\dots,u_{\lceil \frac{k}{2}\rceil+1})$ is the shortest path between $u$ and $w$ in $G$ and such a path has length $(k+1)-(\lceil \frac{k}{2}\rceil+1)=\lfloor \frac{k}{2} \rfloor=i-2$. Since $k\geq 5$ (given that $k\geq 4$ and $k$ is odd), we get $y(w)-y(u)=i-1=\lceil \frac{k}{2}\rceil\leq 2 \lfloor \frac{k}{2} \rfloor =2\cdot \delta_G(u,w)$. Also, we have $x(w)-x(u)=\sum_{j=1}^{\lceil \frac{k}{2}\rceil} s_j\leq \sum_{j=1}^{\lceil \frac{k}{2}\rceil} b_j<\lceil \frac{k}{2}\rceil \cdot b < 2b \cdot \lfloor \frac{k}{2} \rfloor=2b\cdot \delta_G(u,w)$. 

\item Third, suppose that $w=u_i$, for some $\lceil\frac{k}{2}\rceil+2\leq i\leq k$. Note that $\delta_G(u,w)=k-i+1$, since $(u,u_k,u_{k-1},\dots,u_i=w)$ is the shortest path between $u$ and $w$ in $G$. By construction, we have $y(w)-y(u)=k-i+1<2(k-i+1)=2\cdot \delta_G(u,w)$. Also, we have $x(w)-x(u)=\sum_{j=i-\lfloor\frac{k}{2}\rfloor}^{\lceil\frac{k}{2}\rceil} s_j\leq \sum_{j=i-\lfloor\frac{k}{2}\rfloor}^{\lceil\frac{k}{2}\rceil} b_j<(k-i+1) \cdot b < 2b \cdot (k-i+1)=2b\cdot \delta_G(u,w)$.
\item Fourth, suppose that $w$ belongs to a graph $G_i$, for some $1\leq i\leq \lceil \frac{k}{2}\rceil-1$, with $w\notin \{u_i,u_{i+1}\}$. Then $\delta_G(u,w)=i-1+\delta_{G_i}(u_i,w)$, since the shortest path between $u$ and $w$ in $G$ consists of $(u=u_1,u_2,\dots,u_i)$, together with a shortest path between $u_i$ and $w$ in $G_i$. By induction, we have $y(w)-y(u_i)\leq 2\cdot \delta_{G_i}(u_i,w)$. Thus, we have $y(w)-y(u)=(y(u_i)-y(u))+(y(w)-y(u_i))\leq i-1 + 2\cdot \delta_{G_i}(u_i,w)\leq 2(i-1+\delta_{G_i}(u_i,w))=2\delta_G(u,w)$. Also, by induction, we have $x(w)-x(u_i)\leq 2b\cdot \delta_{G_i}(u_i,w)$. Thus, we have $x(w)-x(u)=(x(u_i)-x(u))+(x(w)-x(u_i))\leq \sum_{j=1}^{i-1} s_j + 2b\cdot \delta_{G_i}(u_i,w)\leq \sum_{j=1}^{i-1} b_j + 2b\cdot \delta_{G_i}(u_i,w) <(i-1) \cdot b + 2b\cdot \delta_{G_i}(u_i,w) \leq 2b(i-1+\delta_{G_i}(u_i,w))=2b\cdot\delta_G(u,w)$.
\item Fifth, suppose that $w$ belongs to $G_i$, with $i=\lceil \frac{k}{2}\rceil$. If a shortest path between $u$ and $w$ passes through $u_i$, then again $\delta_G(u,w)=i-1+\delta_{G_i}(u_i,w)$ and the proof that the width and height properties are satisfied is the same as in the fourth case. However, if $k$ is odd, then the shortest path between $u$ and $w$ might not pass through $u_i$. Indeed, the path $(u,u_k,u_{k-1},\dots,u_{\lceil \frac{k}{2}\rceil+1})$ has the same length as the path $(u,u_2,u_3,\dots,u_{\lceil \frac{k}{2}\rceil})$, and the shortest path between $u_{\lceil \frac{k}{2}\rceil+1}$ and $w$ in $G_{\lceil \frac{k}{2}\rceil}$ might be shorter than the one between $u_{\lceil \frac{k}{2}\rceil}$ and $w$. However, since the edge $(u_{\lceil \frac{k}{2}\rceil},u_{\lceil \frac{k}{2}\rceil+1})$ belongs to $G$, it follows that the shortest path between $u_{\lceil \frac{k}{2}\rceil}$ and $w$ has length smaller than or equal to one plus the length of the shortest path between $u_{\lceil \frac{k}{2}\rceil+1}$ and $w$. Thus, if the shortest path between $u$ and $w$ does not pass through $u_i$, we have $\delta_G(u,w)=\lceil \frac{k}{2}\rceil-2+\delta_{G_{\lceil \frac{k}{2}\rceil}}(u_{\lceil \frac{k}{2}\rceil},w))$; note the $-2$, which replaces the $-1$ from the case in which the shortest path from $u$ to $w$ passes through $u_i$. Since $k\geq 5$ (given that $k\geq 4$ and $k$ is odd), we get $y(w)-y(u)\leq \lceil \frac{k}{2}\rceil-1 + 2\cdot \delta_{G_{\lceil \frac{k}{2}\rceil}}(u_{\lceil \frac{k}{2}\rceil},w)\leq 2(\lceil \frac{k}{2}\rceil-2+\delta_{G_{\lceil \frac{k}{2}\rceil}}(u_{\lceil \frac{k}{2}\rceil},w))=2\delta_G(u,w)$. For the same reason, we get $x(w)-x(u)<(\lceil \frac{k}{2}\rceil-1) \cdot b + 2b\cdot \delta_{G_{\lceil \frac{k}{2}\rceil}}(u_{\lceil \frac{k}{2}\rceil},w) \leq 2b({\lceil \frac{k}{2}\rceil}-1+\delta_{G_{\lceil \frac{k}{2}\rceil}}(u_{\lceil \frac{k}{2}\rceil},w))=2b\cdot\delta_G(u,w)$.
\item Finally, suppose that $w$ belongs to a graph $G_i$, for some $\lceil\frac{k}{2}\rceil+1\leq i\leq k-1$, with $w\notin \{u_i,u_{i+1}\}$. Then we have $\delta_G(u,w)=k-i+\delta_{G_i}(u_{i+1},w)$, since the shortest path between $u$ and $w$ in $G$ consists of $(u,u_k,u_{k-1},\dots,u_{i+1})$, together with a shortest path between $u_{i+1}$ and $w$ in $G_i$. By induction, we have $y(w)-y(u_i)\leq 2\cdot \delta_{G_i}(u_{i+1},w)$. Thus, we have $y(w)-y(u)=(y(u_{i+1})-y(u))+(y(w)-y(u_{i+1}))\leq k-i + 2\cdot \delta_{G_i}(u_{i+1},w)\leq 2(k-i+\delta_{G_i}(u_{i+1},w))=2\delta_G(u,w)$. Also, by induction, we have $x(w)-x(u_{i+1})\leq 2b\cdot \delta_{G_i}(u_{i+1},w)$. Thus, we have $x(w)-x(u)=(x(u_{i+1})-x(u))+(x(w)-x(u_{i+1}))\leq \sum_{j=i+1}^{k} s_{j-\lfloor \frac{k}{2}\rfloor} + 2b\cdot \delta_{G_i}(u_{i+1},w)\leq \sum_{j=i+1}^{k} b_{j-\lfloor \frac{k}{2}\rfloor} + 2b\cdot \delta_{G_i}(u_{i+1},w) <(k-i) \cdot b + 2b\cdot \delta_{G_i}(u_{i+1},w) \leq 2b(k-i+\delta_{G_i}(u_{i+1},w))=2b\cdot\delta_G(u,w)$.
\end{itemize}
This concludes the proof that $\Gamma$ is a $(u,v,a,b)$-top-drawing satisfying the width and height properties.

We now proceed to the construction of a $(u,v,\Sigma,a,b)$-side-drawing. The construction distinguishes three cases, namely the one in which $k$ is even, the one in which $k$ is odd and $|\Sigma|=1$, and the one in which $k$ is odd and $|\Sigma|=2$. In all three cases, we represent $(u,v)$ as the segment $\sigma$, with $v$ to the left of $u$; recall that $\sigma$ is the unique $s\times 1$-segment in $\Sigma$ if $|\Sigma|=1$, and is a $(z_1+z_2)\times 2$-segment if $\Sigma$ contains a $z_1\times 1$-segment and a $z_2\times 1$-segment.

\begin{figure}[htb]
    \begin{subfigure}[b]{0.9\textwidth}
      \centering
      \includegraphics[page=7,scale=.9]{Figures/diameter.pdf}
      \subcaption{In this example, we have $k=6$.}
      \label{fig:side-drawing-construction-even-1}
    \end{subfigure}
    \\
    \begin{subfigure}[b]{0.9\textwidth}
      \centering
      \includegraphics[page=8,scale=.9]{Figures/diameter.pdf}
      \subcaption{In this example, we have $k=4$.}
      \label{fig:side-drawing-construction-even-2}
    \end{subfigure}
    \caption{Construction of a $(u,v,\Sigma,a,b)$-top-drawing if $k$ is even and $|\Sigma|=1$ (a) or $|\Sigma|=2$ (b).}
    \label{fig:side-drawing-construction}
\end{figure}

\begin{itemize}
\item First, if $k$ is even, as in \cref{fig:side-drawing-construction-even-1,fig:side-drawing-construction-even-2}, then for $i=1,\dots,\frac{k}{2}-1$, recursively construct a $(u_i,u_{i+1},a_i,b_i)$-top-drawing $\Gamma_i$ of $G_i$. The drawings $\Gamma_1,\Gamma_2,\dots,\Gamma_{\frac{k}{2}-1}$ are translated so that the placement of $u_1$ is the same as in $\sigma$, and the placement of $u_i$ coincides in $\Gamma_i$ and $\Gamma_{i-1}$. Recall that, by definition of top-drawing, for $i=1,\dots,\frac{k}{2}-1$, we have that the edge $(u_i,u_{i+1})$ is represented by an~$s_i\times 1$-segment with~$a_i\leq s_i\leq b_i$. Next, recursively construct a $(u_{\frac{k}{2}},u_{\frac{k}{2}+1},\Sigma_{\frac{k}{2}},a_{\frac{k}{2}},b_{\frac{k}{2}})$-side-drawing $\Gamma_{\frac{k}{2}}$ of $G_{\frac{k}{2}}$, where $\Sigma_{\frac{k}{2}}$ contains only the segment~$\sigma$, and translate it so that the placement of $u_{\frac{k}{2}}$ coincides in $\Gamma_{\frac{k}{2}}$ and $\Gamma_{{\frac{k}{2}}-1}$. Also, for $i=\frac{k}{2}+1,\dots,k-1$, recursively construct a $(u_i,u_{i+1},\Sigma_i,a_i,b_i)$-side-drawing $\Gamma_i$ of $G_i$, where $\Sigma_i$ contains only the segment $\sigma_{i-\frac{k}{2}}$, and translate it so that the placement of $u_i$ coincides in $\Gamma_i$ and $\Gamma_{i-1}$. 
\begin{figure}[htb]
      \centering
      \includegraphics[page=9,scale=.9]{Figures/diameter.pdf}
    \caption{Construction of a $(u,v,\Sigma,a,b)$-top-drawing if $k$ is odd and $|\Sigma|=1$. In this example, we have $k=5$.}
    \label{fig:side-drawing-construction-odd-1}
\end{figure}
\item Second, if $k$ is odd and $|\Sigma|=1$, as in \cref{fig:side-drawing-construction-odd-1}, for $i=1,\dots,\frac{k-1}{2}$, recursively construct a $(u_i,u_{i+1},a_i,b_i)$-top-drawing $\Gamma_i$ of $G_i$. The drawings $\Gamma_1,\Gamma_2,\dots,\Gamma_{\frac{k-1}{2}}$ are translated so that the placement of $u_1$ is the same as in $\sigma$, and the placement of $u_i$ coincides in $\Gamma_i$ and $\Gamma_{i-1}$. Recall that, by definition of top-drawing, for $i=1,\dots,\frac{k-1}{2}$, we have that the edge $(u_i,u_{i+1})$ is represented by an~$s_i\times 1$-segment with~$a_i\leq s_i\leq b_i$. Next, recursively construct a $(u_{\frac{k+1}{2}},u_{\frac{k+3}{2}},\Sigma_{\frac{k+1}{2}},a_{\frac{k+1}{2}},b_{\frac{k+1}{2}})$-side-drawing of $G_{\frac{k+1}{2}}$, where $\Sigma_{\frac{k+1}{2}}$ contains $\sigma$ and an $s_1\times 1$-segment, and translate it so that the placement of $u_{\frac{k+1}{2}}$ coincides in $\Gamma_{\frac{k+1}{2}}$ and $\Gamma_{\frac{k-1}{2}}$. Also, for $i=\frac{k+3}{2},\dots,k-1$, recursively construct a $(u_i,u_{i+1},\Sigma_i,a_i,b_i)$-side-drawing of $G_i$, where $\Sigma_i$ only contains the segment $\sigma_{i-\frac{k-1}{2}}$, and translate it so that the placement of $u_i$ coincides in $\Gamma_i$ and $\Gamma_{i-1}$. 
\begin{figure}[htb]
      \centering
      \includegraphics[page=10,scale=.9]{Figures/diameter.pdf}
    \caption{Construction of a $(u,v,\Sigma,a,b)$-top-drawing if $k$ is odd and $|\Sigma|=2$. In this example, we have $k=5$.}
    \label{fig:side-drawing-construction-odd-2}
\end{figure}
\item Finally, if $k$ is odd and $|\Sigma|=2$, as in \cref{fig:side-drawing-construction-odd-2}, for $i=1,\dots,\frac{k-3}{2}$, recursively construct a $(u_i,u_{i+1},a_i,b_i)$-top-drawing $\Gamma_i$ of $G_i$. The drawings $\Gamma_1,\Gamma_2,\dots,\Gamma_{\frac{k-3}{2}}$ are translated so that the placement of $u_1$ is the same as in $\sigma$, and the placement of $u_i$ coincides in $\Gamma_i$ and $\Gamma_{i-1}$. Recall that, by definition of top-drawing, for $i=1,\dots,\frac{k-3}{2}$, we have that the edge $(u_i,u_{i+1})$ is represented by an~$s_i\times 1$-segment with~$a_i\leq s_i\leq b_i$. Next, recursively construct a $(u_{\frac{k-1}{2}},u_{\frac{k+1}{2}},\Sigma_{\frac{k-1}{2}},a_{\frac{k-1}{2}},b_{\frac{k-1}{2}})$-side-drawing of $G_{\frac{k-1}{2}}$, where $\Sigma_{\frac{k-1}{2}}$ only contains a $z_1\times 1$-segment, and translate it so that the placement of $u_{\frac{k-1}{2}}$ coincides in $\Gamma_{\frac{k-3}{2}}$ and $\Gamma_{\frac{k-1}{2}}$. Also, recursively construct a $(u_{\frac{k+1}{2}},u_{\frac{k+3}{2}},\Sigma_{\frac{k+1}{2}},a_{\frac{k+1}{2}},b_{\frac{k+1}{2}})$-side-drawing of $G_{\frac{k+1}{2}}$, where $\Sigma_{\frac{k+1}{2}}$ only contains a $z_2\times 1$-segment, and translate it so that the placement of $u_{\frac{k+1}{2}}$ coincides in $\Gamma_{\frac{k-1}{2}}$ and $\Gamma_{\frac{k+1}{2}}$. For $i=\frac{k+3}{2},\dots,k-1$, recursively construct a $(u_i,u_{i+1},\Sigma_i,a_i,b_i)$-side-drawing of $G_i$, where $\Sigma_i$ only contains the segment $\sigma_{i-\frac{k-1}{2}}$, and translate it so that the placement of $u_i$ coincides in $\Gamma_i$ and $\Gamma_{i-1}$. 
\end{itemize} 

This completes the construction of $\Gamma$. The proof that $\Gamma$ satisfies the required properties is very similar to the one for a top-drawing and we hence only highlight the differences. First, note that, in all three cases, the edge $(u,v)$ is represented by $\sigma$, with $v$ to the left of $u$, by construction, thus satisfying Property S1. The proof that $\Gamma$ is an internally-strictly-convex grid drawing satisfying Property S2 follows the one for a top-drawing almost verbatim, expect for a change in the indices, given that the subgraphs $G_i$ which have a top-drawing and are drawn ``to the left'' of the face $f$ are those with $1\leq i\leq \lfloor \frac{k}{2}\rfloor$ (in the first two cases) or with $1\leq i\leq \lfloor \frac{k}{2}\rfloor-1$ (in the third case), rather than those with $1\leq i\leq \lceil \frac{k}{2}\rceil$ as in the construction of a top-drawing, given that here the edge $(u,v)$ is the part of the path on the boundary of $f$ drawn with decreasing slopes, rather than of the path on the boundary of $f$ drawn with increasing slopes. The width and height properties are again satisfied because each segment representing an edge on the boundary of $f$ has $y$-extension at most $2$ and $x$-extension at most $2b$. Similarly to the construction of a top-drawing, if $k$ is odd and $\Sigma=1$, the shortest path from $v$ to a vertex $w$ in $G_i$ with $i=\frac{k-1}{2}$ might not pass through $u_{\frac{k-3}{2}}$ and use the route $(v,u_{k-1},u_{k-2},\dots,u_{i+1})$ instead. Here we also have a symmetric case. Namely, if $k$ is odd and $\Sigma=2$, the shortest path from $v$ to a vertex $w$ in $G_i$ with $i=\frac{k-1}{2}$ might not pass through $u_{\frac{k+1}{2}}$ and use the route $(v,u_1,u_2,\dots,u_i)$ instead. However, in both cases, the same arguments used as in the construction of a top-drawing ensure that the drawing satisfies the width and height properties. This concludes the induction and hence the proof of \cref{th:upper-diameter}.

The following is a consequence of \cref{th:upper-diameter}.

\begin{theorem} \label{th:upper-diameter-outer1}
Every~$n$-vertex outer-$1$-plane graph with diameter $d$ admits an embedding-preserving straight-line grid drawing in $O(nd^2)$ area.
\end{theorem}

\begin{proof}
Let $G$ be an~$n$-vertex outer-$1$-plane graph with diameter $d$. We first prove that the outer frame $G'$ of $G$ has diameter at most $2d$. Note that the framed augmentation $F$ of $G$ has diameter at most $d$, since it is obtained by adding edges to $G$. Also, for every path~$P=(u_1,u_2,\dots,u_k)$ in~$F$, there exists a path in $G'$ whose length is at most twice the length of $P$. Indeed, we construct a connected subgraph~$H'$ of~$G'$ that contains~$u_1$ and~$u_k$ and that has at most $2k-2$ edges, as follows. For~$i=1,\dots,k-1$, we add the edge~$(u_i,u_{i+1})$ to~$H'$ if it belongs to~$G'$. Otherwise,~$(u_i,u_{i+1})$ has been removed from $F$ since it crosses an edge $e$. This implies that $G'$ contains a $4$-cycle on the end-vertices of~$(u_i,u_{i+1})$ and $e$, hence $G'$ contains a length-$2$ path between $u_i$ and $u_{i+1}$. We add such a path to~$H'$. The number of edges of~$H'$ is at most twice the length of~$P$, hence at most $2k-2$. Thus,~$H'$ contains a path from~$u_1$ to~$u_k$ whose length is at most twice the length of~$P$.

By \cref{th:upper-diameter}, we have that $G'$ admits an internally-strictly-convex drawing $\Gamma'$ in $O(nd^2)$ area. Since the internal faces of $G'$ are represented as strictly-convex polygons in $\Gamma'$, re-introducing the removed crossing edges in $\Gamma'$ as straight-line segments, and removing the edges of $G'$ that do not belong to $G$ results in an embedding-preserving straight-line grid drawing of $G$ in $O(nd^2)$ area. 
\end{proof}

\section{Conclusions}
\label{sec:conclusions}

In this paper we have proved the first sub-quadratic area upper bound for internally-convex grid drawings of outerplanar graphs. We also presented an algorithm to construct internally-strictly-convex grid drawings of outerpaths whose area is asymptotically optimal, with respect to the number of vertices of the graph and the maximum size of an internal face. Finally, we presented an algorithm to construct internally-strictly-convex grid drawings of outerplanar graphs whose area is asymptotically optimal, with respect to the number of vertices of the graph and the diameter. 

Several intriguing questions are left open by our research:
\begin{itemize}
    \item A first, natural, research direction is to narrow the gap between the $O(n^{1.5})$ upper bound and the $\Omega(n)$ lower bound for internally-convex grid drawings of~$n$-vertex outerplanar graphs. In particular, any super-linear lower bound would be an important step towards proving an analogous result for (not necessarily convex) planar straight-line grid drawings.
    \item Concerning internally-strictly-convex grid drawings of (general) $n$-vertex outerplane graphs whose faces have size at most $k$, for which we proved an $\Omega(nk^2)$ lower bound, we are not aware of any upper bound better than $O(n^3)$~\cite{Andrews1963}. Tightening this gap is a nice goal. In particular, we ask: Does every $n$-vertex outerplane graph whose faces have size at most $k$ admit an internally-strictly-convex grid drawing in $O(nk^2)$ area? We could only answer this question affirmatively for outerpaths.
    \item The case $k=4$, for which there is an $O(n^2)$ upper bound (a consequence of results in~\cite{Barany2006}), is especially interesting. Indeed, any upper bound would translate to an upper bound for the area requirements of straight-line drawings of outer-1-planar graphs.
    \item Finally, we have observed that, for (non-strictly) convex grid drawings of~$n$-vertex outerplane graphs (in which the outer face is also required to be convex), 
    $\Theta(n^3)$ is a tight bound for the area requirements. However, if the maximum size of an internal face is bounded, our lower bound does not hold anymore. We believe that an $\Omega(n \log n)$ lower bound can be proved, as the convexity of the outer face would require a dimension to have $\Omega(n)$ length and as there exist $n$-vertex trees that require $\Omega(\log n)$ length in both dimensions in any planar straight-line drawing (see, e.g.,~\cite{DBLP:journals/jgaa/FelsnerLW03}). However, this lower bound is far from the $O(n^3)$ upper bound and it would be interesting to close this gap.
\end{itemize}

%We conclude by mentioning that we believe that the techniques we developed in order to construct internally-strictly-convex grid drawings of outerpaths could lead to non-trivial area bounds for outerplanar graphs whose weak dual tree has bounded diameter. We plan to investigate the truth of this intuition in the near future.

%Convex drawings (even the outer face) of outerplane graphs with small faces?  and $O(n^3)$ seem to be the best. Reference Theorem 1 with parameterization on the face size. Does every outer-1-planar graph admit a straight-line drawing in $o(n^4)$ area. Note that an upper bound of~$O(n^4)$ can be achieved by first making the outer-1-planar graph maximal and then drawing its planar skeleton (whose internal faces all have degree four) by the algorithm of B{\'{a}}r{\'{a}}ny and Rote~\cite{Barany2006}?

\bibliographystyle{plainurl}
\bibliography{references}

\newpage
\appendix

% \section{Full Proofs from \cref{se:non-strict}}

% \CaseOne*

% \caseTwo*

\end{document}